\documentclass[sn-mathphys,Numbered]{sn-jnl}% Math and Physical Sciences Reference Style
\usepackage{subcaption}
\usepackage{graphicx}%
\usepackage{multirow}%
\usepackage{amsmath,amssymb,amsfonts}%
\usepackage{amsthm}%
\usepackage{mathrsfs}%
\usepackage[title]{appendix}%
\usepackage{algorithm}
\usepackage[noend]{algpseudocode}
\usepackage{xcolor}%
\usepackage{textcomp}%
\usepackage{manyfoot}%
\usepackage{booktabs}%
\usepackage{enumitem}
\usepackage{bbm}
\usepackage{mathptmx}
\usepackage{ragged2e,booktabs,tabularx,adjustbox}

\newcommand{\scrF}{\mathscr{F}}

\newtheorem{theorem}{Theorem}%  meant for continuous numbers
\newtheorem{corollary}[theorem]{Corollary}

\begin{document}

\title[Article Title]{European Option Pricing under Generalized Tempered Stable Process: Empirical Analysis}


\author*[1]{\fnm{Aubain} \sur{Nzokem}}\email{hilaire77@gmail.com}


%\affil*[1]{\orgdiv{Department of Mathematics \& Statistics}, \orgname{York University}, \orgaddress{\street{4700 Keele Street}, \city{Toronto}, \postcode{M3J1P3}, \state{ON}, \country{Canada}}}

%\affil[3]{\orgdiv{Department}, \orgname{Organization}, \orgaddress{\street{Street}, \city{City}, \postcode{610101}, \state{State}, \country{Country}}}


\abstract{The paper investigates the performance of the European option price when the log asset price follows a  rich class of Generalized Tempered Stable (GTS) distribution. The GTS distribution is an alternative to Normal distribution and $\alpha$-stable distribution for modeling asset return and many physical and economic systems. The data used in the option pricing computation comes from fitting the GTS distribution to the underlying S\&P 500 Index return distribution. The Esscher transform method shows that the GTS distribution preserves its structure. The extended Black-Scholes formula and the Generalized Black-Scholes Formula are applied in the study. The 12-point rule Composite Newton-Cotes Quadrature and the Fractional Fast Fourier (FRFT) algorithms were implemented and they yield the same European option price at two decimal places. Compared to the option price under the GTS distribution, the Black-Scholes (BS) model is underpriced for the Near-The-Money (NTM) and the in-the-money (ITM) options. However, the BS model and GTS European options yield the same option price for the deep out-of-the-money (OTM) and the deep-in-the-money (ITM) options.}




\maketitle
 \section{Introduction}\label{sec1}
\noindent
Black-Scholes (BS) model \cite{black1973} is considered the cornerstone of the option pricing theory. The model relies on the fundamental assumption that the asset returns have a normal distribution with a known mean and variance. However, based on empirical studies, the BS model model is inconsistent with a set of well-established stylized features\cite{Cont_2001}. The $\alpha$-stable distribution has been proposed \cite{ken1999levy,nolan2020univariate} as an alternative to the normal distribution for modeling asset return and many types of physical and economic systems. In addition to the theoretical and empirical arguments for modeling with a stable distribution, the most important argument is that the Central Limit theorem can be generalized by the stable distribution \cite{rachev2011financial,carr2007self,nzokem2024self}. Although the stable distribution allows for varying degrees of tail heaviness and skewness, it has two major drawbacks \cite{nolan2020univariate, borak2005stable}: firstly, some lack of closed formulas for density and cumulative distribution functions; secondly, most of the moments of the stable distribution are infinite. An infinite variance of the asset return leads to an infinite price for derivative instruments such as options. The Generalized Tempered Stable (GTS) distribution was developed to overcome the shortcomings of the $\alpha$-stable distribution speciﬁcally in modeling high-frequency asset returns. The tails of the GTS distribution for asset returns are heavier than the normal distribution but thinner than the stable distribution \cite{grabchak2010financial, nolan2020univariate}. See \cite{kuchler2013tempered,rachev2011financial} For more details on Tempered Stable distribution.\\
\noindent
There has been considerable research on the stochastic volatility model for option pricing. In the current study, we contribute to the literature by empirically investigating the  European option pricing under GTS distribution. The GTS distributions form a seven-parameter family of infinitely divisible distributions, which cover several well-known subclasses like Variance Gamma distributions \cite{madan1998variance, nzokem2021fitting, Nzokem_Montshiwa_2023}, bilateral Gamma distributions \cite{kuchler2013tempered,Lorenzo2014,Dong2024} and CGMY distributions \cite{carr2003stochastic, risks10080148, Asmussen}. The S\&P 500 index historical data was used to fit the GTS distribution to the underlying data return distribution. The methodology is based on two computational algorithms \cite{mca29030044,nzokem2023enhanced}: the Fractional Fast Fourier (FRFT) \cite{cherubini2010fourier,eberlein2010analysis} and twelve-point rule Composite Newton--Cotes Quadrature \cite{nzokem_2021b,nzokem_2021}.\\
\noindent
We organize the paper as follows. We briefly present the GTS distribution's theoretical framework in the next section. In section 3, we estimate the seven parameters of the GTS distribution from the S\&P 500 index historical data. In section 4, the Equivalent Martingale Measure (EMM) is identified and computed; the Extended Black-Scholes and the Generalized Black-Scholes Formulas are provided. And the last section presents the empirical analysis of the European option price computations.

%\newpage
\section{ Generalized Tempered Stable (GTS) Process}

\noindent
The L\'evy measure of the Generalized Tempered Stable (GTS) distribution ($V(dx)$) is defined (\ref{eq:l3}) as a product of a tempering function ($q(x)$) (\ref{eq:l1}) and a L\'evy measure of the $\alpha$-stable distribution ($V_{stable}(dx)$)(\ref{eq:l2}).
 \begin{align}
q(x) &= e^{-\lambda_{+}x} \boldsymbol{1}_{x>0} + e^{-\lambda_{-}|x|} \boldsymbol{1}_{x<0} \label{eq:l1}\\
V_{stable}(dx) &=\left(\frac{\alpha_{+}}{x^{1+\beta_{+}}} \boldsymbol{1}_{x>0} + \frac{\alpha_{-}}{|x|^{1+\beta_{-}}} \boldsymbol{1}_{x<0}\right) dx \label{eq:l2}\\
V(dx) =q(x)V_{stable}(dx)&=\left(\frac{\alpha_{+}e^{-\lambda_{+}x}}{x^{1+\beta_{+}}} \boldsymbol{1}_{x>0} + \frac{\alpha_{-}e^{-\lambda_{-}|x|}}{|x|^{1+\beta_{-}}} \boldsymbol{1}_{x<0}\right) dx \label{eq:l3}
 \end{align}
\noindent 
where $0\leq \beta_{+}\leq 1$, $0\leq \beta_{-}\leq 1$, $\alpha_{+}\geq 0$, $\alpha_{-}\geq 0$, $\lambda_{+}\geq 0$ and  $\lambda_{-}\geq 0$. These parameters play an essential role in the Levy process. $\beta_{+}$ and $\beta_{-}$ are the indexes of stability bounded below by 0 and above by 2 \cite{borak2005stable}. They capture the peakedness of the distribution in a similar way as the  $\beta$-stable distribution, but the distribution tails are tempered. If $\beta$ increases (decreases), then the peakedness decreases (increases). $\alpha_{+}$ and $\alpha_{-}$ are the scale parameter, also called the process intensity \cite{boyarchenko2002non}, they determine the arrival rate of jumps for a given size. $\lambda_{+}$ and $\lambda_{-}$ control the decay rate on the positive and negative tails. Additionally, $\lambda_{+}$ and $\lambda_{-}$ are also skewness parameters. If $\lambda_{+}>\lambda_{-}$ ($\lambda_{+}<\lambda_{-}$), then the distribution is skewed to the left (right), and if $\lambda_{+}=\lambda_{-}$, then it is symmetric \cite{rachev2011financial,fallahgoul2019quantile}.
 \\
\noindent
The GTS distribution can be denoted by $X\sim GTS(\textbf{$\beta_{+}$}, \textbf{$\beta_{-}$}, \textbf{$\alpha_{+}$},\textbf{$\alpha_{-}$}, \textbf{$\lambda_{+}$}, \textbf{$\lambda_{-}$})$ and $X=X_{+} - X_{-}$ with $X_{+} \geq 0$, $X_{-}\geq 0$. $X_{+}\sim TS(\textbf{$\beta_{+}$}, \textbf{$\alpha_{+}$},\textbf{$\lambda_{+}$})$ and  $X_{-}\sim TS(\textbf{$\beta_{-}$}, \textbf{$\alpha_{-}$},\textbf{$\lambda_{-}$})$. 
\begin{align}
  \int_{-\infty}^{+\infty} V(dx) =\begin{cases}
  +\infty  &\text{if }{\beta_{+}\ge 0\vee\beta_{-} \ge 0}   \\
  \alpha_{+}{\lambda_{+}}^{\beta_{+}}\Gamma(-\beta_{+}) +  \alpha_{-}{\lambda_{-}}^{\beta_{-}}\Gamma(-\beta_{-}) &\text{if }{\beta_{+} < 0\wedge\beta_{-} < 0} \end{cases} \label{eq:l4}
     \end{align} 
From (\ref{eq:l4}), it results that when $\beta_{+} < 0$, TS(\textbf{$\beta_{+}$}, \textbf{$\alpha_{+}$},\textbf{$\lambda_{+}$}) is of finite activity and can be written as a Compound Poisson process on the right side ($X_{+}$). we have similar pattern when $\beta_{-} < 0$. However, when $0\le \beta_{+} \le 1$, $X_{+}$ is an infinite activity process with infinite jumps in any given time interval. We have a similar pattern when $0 \le \beta_{-} \le 1$. In addition to the infinite activities process, we have 
\begin{align}
  \int_{-\infty}^{+\infty} min(1,|x|)V(dx) <+ \infty \label{eq:l5} 
\end{align}
\noindent
By adding the location parameter, the GTS distribution becomes GTS(\textbf{$\mu$}, \textbf{$\beta_{+}$}, \textbf{$\beta_{-}$}, \textbf{$\alpha_{+}$},\textbf{$\alpha_{-}$}, \textbf{$\lambda_{+}$}, \textbf{$\lambda_{-}$}) and we have:
 \begin{align}
Y=\mu + X=\mu + X_{+} - X_{-} \quad \quad  Y\sim GTS(\mu, \beta_{+}, \beta_{-}, \alpha_{+}, \alpha_{-},\lambda_{+}, \lambda_{-}) \label {eq:l6}
  \end{align}

\begin{theorem}\label{lem5} \ \\ 
Consider a variable $Y \sim GTS(\textbf{$\mu$}, \textbf{$\beta_{+}$}, \textbf{$\beta_{-}$}, \textbf{$\alpha_{+}$},\textbf{$\alpha_{-}$}, \textbf{$\lambda_{+}$}, \textbf{$\lambda_{-}$})$, the characteristic exponent can be written
  \begin{align}
\Psi(\xi)=\mu\xi i + \alpha_{+}\Gamma(-\beta_{+})\left((\lambda_{+} - i\xi)^{\beta_{+}} - {\lambda_{+}}^{\beta_{+}}\right) + \alpha_{-}\Gamma(-\beta_{-})\left((\lambda_{-} + i\xi)^{\beta_{-}} - {\lambda_{-}}^{\beta_{-}}\right) \label {eq:l7}
  \end{align}
\end{theorem} 
\noindent
See  \cite{jrfm17120531} for theorem \ref{lem5} proof \\

\begin{corollary}\label{lem4} \ \\
Let $Y=\left(Y_{t}\right)$ be a L\'evy process on $\mathbb{R^{+}}$ generated by $GTS(\textbf{$\mu$}, \textbf{$\beta_{+}$}, \textbf{$\beta_{-}$}, \textbf{$\alpha_{+}$},\textbf{$\alpha_{-}$}, \textbf{$\lambda_{+}$}, \textbf{$\lambda_{-}$})$. \\
$\forall t \in \mathbb{R^{+}}$, 
 \begin{align}
Y_{t} \sim GTS(t\mu, \beta_{+}, \beta_{-}, t\alpha_{+}, t\alpha_{-}, \lambda_{+}, \lambda_{-}) \label {eq:l8}
  \end{align}
\end{corollary}
See  \cite{jrfm17120531} for corollary \ref{lem4} proof 

\section{Parameter Estimation based on daily S\&P 500 price}
\noindent 
The Standard \& Poor’s 500 Composite Stock Price Index, also known as the S\&P 500, is a stock index that tracks the share prices of 500 of the largest public companies in the United States. It is used as a proxy describing the overall health of the stock market or even the U.S. economy. The daily prices were adjusted for splits and dividends. The period spans from January 04, 2010, to June 16, 2023 and the data were extracted from Yahoo Finance. \\

\begin{figure}[ht]
\vspace{-0.6cm}
    \centering
\hspace{-1.4cm}
  \begin{subfigure}[b]{0.45\linewidth}
    \includegraphics[width=\linewidth]{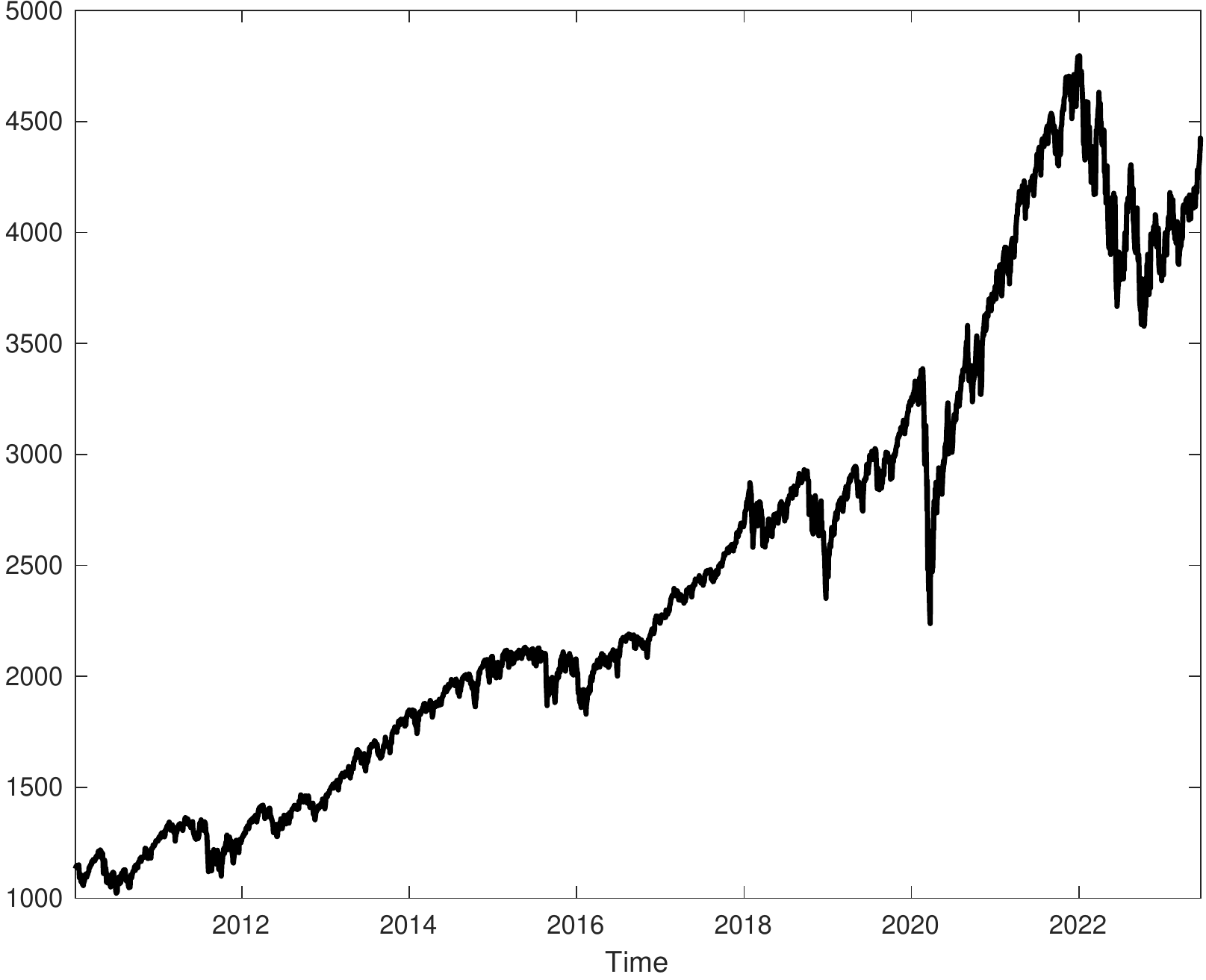}
\vspace{-0.5cm}
     \caption{S\&P 500 daily price}
         \label{fig1}
  \end{subfigure}
\hspace{-0.3cm}
  \begin{subfigure}[b]{0.45\linewidth}
    \includegraphics[width=\linewidth]{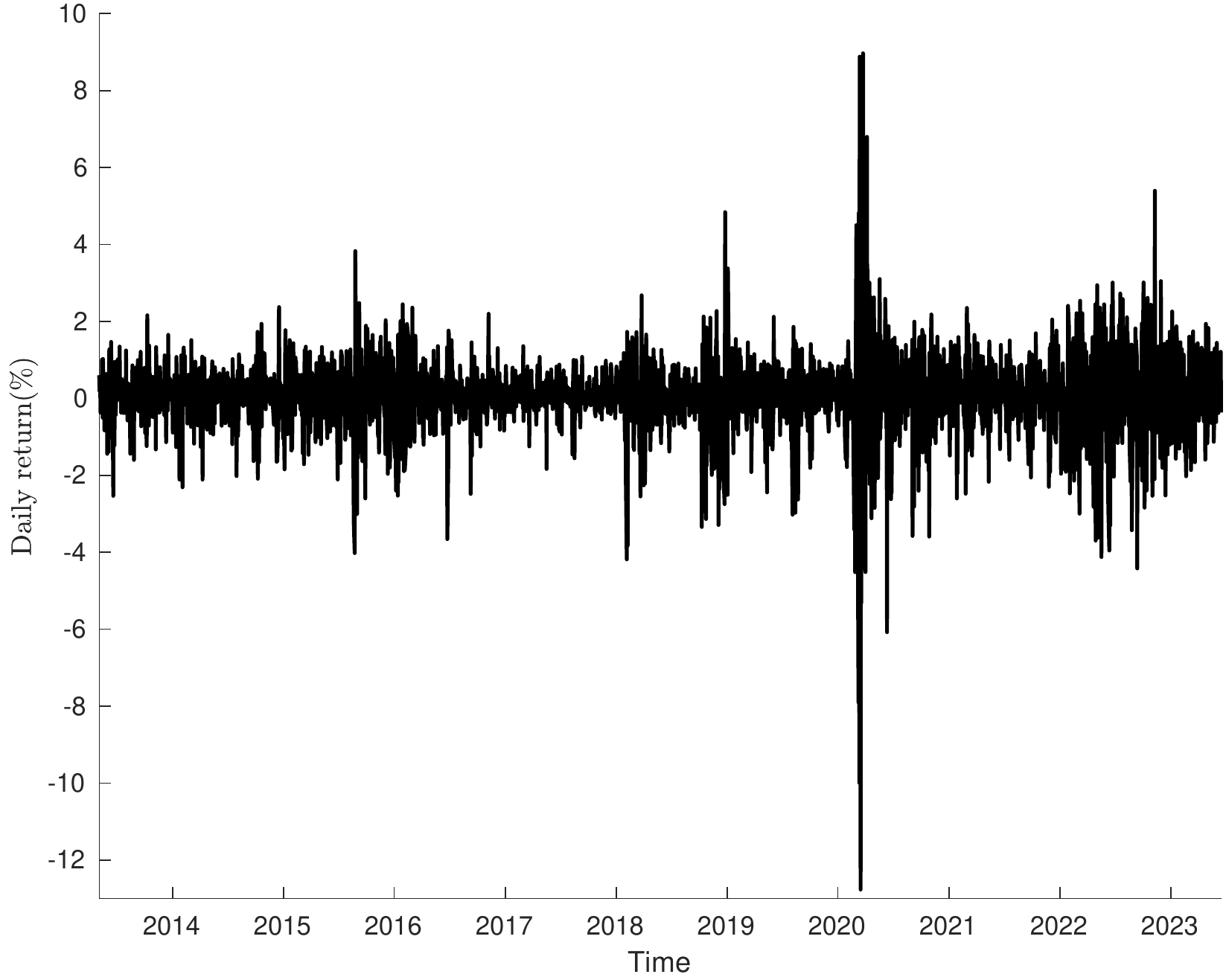}
\vspace{-0.5cm}
     \caption{S\&P 500 daily Return ($\%$)}
         \label{fig2}
          \end{subfigure}
\vspace{-0.5cm}
  \caption{S\&P 500 Index Historical Data}
  \label{fig3}
\vspace{-0.6cm}
\end{figure}
\noindent 
The daily price dynamics are provided in Fig \ref{fig1} and the daily return in Fig \ref{fig2}. As shown in Fig \ref{fig1}, prices have an increasing trend throughout the period, even after the temporal disruption in the first quarter of 2020 by the coronavirus pandemic. The volatility (in Fig \ref{fig2}) of the daily return is higher in the First quarter of 2020 amid the coronavirus pandemic and massive disruptions in the global economy.\\

\noindent
The characteristic function of the GTS process $Y= \{Y_{t}\}$ has the following expression
\begin{align} 
 \vartheta(\xi,t)=E\left[e^{i Y_{t}\xi}\right]=e^{t\Psi(\xi)}  \label {eq:l9}\end{align}
\medskip
The GTS density function ($f$) does not have an explicit closed form or an analytical expression, which makes it challenging to utilize the density function and its derivatives and to perform the Maximum likelihood method. The GTS density function ($f$)  was derived from Fourier transform ($F(f)$) as follows:
\begin{equation}
F[f](\xi) = \vartheta(-\xi,1) \quad \quad  f(y) = \frac{1} {2\pi}\int_{-\infty}^{+\infty}\! F[f](x) e^{iyx}\, \mathrm{d}x  \label {eq:l10}
\end{equation}
\noindent
The Fractional Fourier Transform (FRFT) technique was used to compute the probability density function (\ref {eq:l10}) and its derivatives $\left\{\frac{df(y,V)}{dV_j}\right\}_{1\leq j \leq 7}$ and $\left\{\frac{d^{2}f(y,V)}{dV_{k}dV_{j}}\right\}_{\substack{1\leq k \leq 7\\ 1\leq j \leq 7}}$. See \cite{nzokem2021fitting,nzokem2022,nzokem_2021b} for more details on FRFT methodologies.\\

\noindent 
The results of the GTS Parameter Estimation are reported in Table \ref{tab1}. As expected, the stability index ($\beta$), the process intensity ($\alpha$), and the decay rate ($\lambda$) are all positive. Both stability indexes are less than 1, and capture the high peakedness of the underlying distribution.  As shown in Table \ref{tab1},  S\&P 500 return is a bit left-skewed distribution ($\lambda_{+}>\lambda_{-}$), and the right-side distribution has a high arrival rate of jumps ($\alpha_{+}>\alpha_{-}$). The higher arrival rate of jump on the right-side distribution contributes to the increased nature of the S\&P 500 daily price in Fig \ref{fig1}.

\begin{table}[ht]
\caption{ FRFT Maximum Likelihood GTS Parameter Estimations}
\label{tab1} 
\vspace{-0.3cm}
\centering
\begin{tabular}{@{}lccccccc@{}}
\toprule
\textbf{Model} & \textbf{$\mu$} & \textbf{$\beta_{+}$} & \textbf{$\beta_{-}$} & \textbf{$\alpha_{+}$} & \textbf{$\alpha_{-}$}  & \textbf{$\lambda_{+}$}  & \textbf{$\lambda_{-}$}  \\ \midrule
\textbf{GTS}  & -0.693477 & 0.682290 & 0.242579 & 0.458582 & 0.414443 & 0.822222 & 0.727607  \\ \bottomrule 
\end{tabular}
\end{table} 

\noindent
The Fractional Fourier Transform (FRFT)  was used to estimate $(\mu, \beta_{+}, \beta_{-}, \alpha_{+}, \alpha_{-},\lambda_{+}, \lambda_{-})$ parameter. The Maximum likelihood method was applied to the density function (\ref{eq:l10}) and the second and third derivatives were used to compute the highest eigenvalues of the hessian Matrix and the value of $||\frac{d\log(ML)}{dV}||$.\\

\begin{table}[ht]
\vspace{-0.6cm}
\centering
\caption{GTS Parameters Estimations for S\&P 500 data}
\label{tab2}
\vspace{-0.3cm}
\setlength\extrarowheight{0.05pt} % for a less cramped look
\setlength\tabcolsep{3pt}  
%\resizebox{15cm}{!}{%
\begin{tabular}{ccccccccccc}
\hline
\textbf{$Iterations$} & \textbf{$\mu$} & \textbf{$\beta_{+}$} & \textbf{$\beta_{-}$} & \textbf{$\alpha_{+}$} & \textbf{$\alpha_{-}$} & \textbf{$\lambda_{+}$} & \textbf{$\lambda_{-}$} & \textbf{$Log(ML)$} & \textbf{$||\frac{dLog(ML)}{dV}||$} & \textbf{$Max Eigen$} \\ \hline
1  & -0.5267  & 0.6767  & 0.4366  & 0.4311  & 0.3259  & 0.8501  & 0.6015  & -4664.765  & 289.207  & 48.329 \\
2  & -0.5460  & 0.6706  & 0.4242  & 0.4499  & 0.3481  & 0.8092  & 0.6029  & -4660.216  & 35.985  & -6.043 \\
3  & -0.7108  & 0.6668  & 0.2061  & 0.4696  & 0.4125  & 0.8369  & 0.7368  & -4663.578  & 1082.003  & 449.252 \\
4  & -0.6704  & 0.6689  & 0.1125  & 0.4653  & 0.5003  & 0.8342  & 0.8633  & -4660.527  & 135.685  & 15.110 \\
5  & -0.7398  & 0.6678  & 0.0910  & 0.4830  & 0.4592  & 0.8555  & 0.8132  & -4660.021  & 45.708  & 10.842 \\
6  & -0.6517  & 0.6558  & 0.1967  & 0.4800  & 0.4274  & 0.8525  & 0.7535  & -4659.833  & 46.333  & 11.853 \\
7  & -0.8137  & 0.7200  & 0.2156  & 0.4402  & 0.4195  & 0.7942  & 0.7403  & -4662.482  & 1187.982  & 166.007 \\
8  & -0.7806  & 0.7064  & 0.2295  & 0.4467  & 0.4166  & 0.8036  & 0.7334  & -4659.776  & 85.658  & -3.431 \\
9  & -0.7544  & 0.6991  & 0.2347  & 0.4503  & 0.4154  & 0.8094  & 0.7308  & -4659.194  & 1.074  & -0.799 \\
10  & -0.7534  & 0.6989  & 0.2348  & 0.4504  & 0.4154  & 0.8096  & 0.7307  & -4659.194  & 1.037  & -0.814 \\
11  & -0.7524  & 0.6986  & 0.2349  & 0.4506  & 0.4154  & 0.8098  & 0.7307  & -4659.194  & 1.002  & -0.827 \\
12  & -0.7515  & 0.6983  & 0.2350  & 0.4507  & 0.4154  & 0.8100  & 0.7306  & -4659.194  & 0.969  & -0.840 \\
13  & -0.7497  & 0.6979  & 0.2352  & 0.4509  & 0.4154  & 0.8103  & 0.7306  & -4659.194  & 0.907  & -0.865 \\
14  & -0.7472  & 0.6972  & 0.2355  & 0.4513  & 0.4153  & 0.8109  & 0.7304  & -4659.194  & 0.827  & -0.899 \\
15  & -0.7464  & 0.6970  & 0.2356  & 0.4514  & 0.4153  & 0.8110  & 0.7304  & -4659.194  & 0.802  & -0.909 \\
16  & -0.7456  & 0.6968  & 0.2357  & 0.4515  & 0.4153  & 0.8112  & 0.7304  & -4659.194  & 0.779  & -0.919 \\
17  & -0.7434  & 0.6962  & 0.2360  & 0.4518  & 0.4153  & 0.8116  & 0.7303  & -4659.194  & 0.715  & -0.948 \\
18  & -0.7427  & 0.6960  & 0.2360  & 0.4519  & 0.4153  & 0.8118  & 0.7302  & -4659.193  & 0.696  & -0.957 \\
19  & -0.7376  & 0.6946  & 0.2367  & 0.4525  & 0.4152  & 0.8129  & 0.7300  & -4659.193  & 0.565  & -1.020 \\
20  & -0.7354  & 0.6940  & 0.2369  & 0.4529  & 0.4152  & 0.8133  & 0.7299  & -4659.193  & 0.514  & -1.047 \\
21  & -0.7343  & 0.6937  & 0.2371  & 0.4530  & 0.4151  & 0.8136  & 0.7298  & -4659.193  & 0.477  & -1.063 \\
22  & -0.7277  & 0.6919  & 0.2379  & 0.4539  & 0.4150  & 0.8149  & 0.7295  & -4659.192  & 0.303  & -1.156 \\
23  & -0.7269  & 0.6917  & 0.2380  & 0.4540  & 0.4150  & 0.8151  & 0.7294  & -4659.192  & 0.287  & -1.167 \\
24  & -0.7262  & 0.6915  & 0.2381  & 0.4541  & 0.4150  & 0.8153  & 0.7294  & -4659.192  & 0.272  & -1.177 \\
25  & -0.7074  & 0.6863  & 0.2406  & 0.4567  & 0.4147  & 0.8192  & 0.7284  & -4659.192  & 0.151  & -1.371 \\
26  & -0.7029  & 0.6850  & 0.2412  & 0.4573  & 0.4146  & 0.8202  & 0.7282  & -4659.192  & 0.121  & -1.403 \\
27  & -0.6988  & 0.6838  & 0.2418  & 0.4578  & 0.4145  & 0.8211  & 0.7279  & -4659.191  & 0.079  & -1.428 \\
28  & -0.6935  & 0.6823  & 0.2426  & 0.4586  & 0.4144  & 0.8222  & 0.7276  & -4659.191  & 0.753  & -1.644 \\
29  & -0.6935  & 0.6823  & 0.2426  & 0.4586  & 0.4144  & 0.8222  & 0.7276  & -4659.191  & 0.000  & -1.454 \\
30  & -0.6935  & 0.6823  & 0.2426  & 0.4586  & 0.4144  & 0.8222  & 0.7276  & -4659.191  & 0.000  & -1.454 \\ \hline
\end{tabular}
%}
\end{table}

\noindent
The parameter estimations of the last 30 iteration processes are shown in Table \ref {tab2}; in addition to the seven parameters displayed in the first row, we have added three statistical indicators: the value of the logarithmic of the product of the density functions ($\log(ML)$), the norm of the partial derivative function $||\frac{d\log(ML)}{dV}||$, and the highest value of the Eigenvalues of the Hessian matrix  ($Max Eigen Value$). As shown in the last row of Table \ref {tab2}, the iteration process stops when the estimated parameter converges to a stable parameter with a maximum value $\log(ML)=-4659.1914$,  $||\frac{d\log(ML)}{dV}||=0$  and $Max Eigen Value= -1.45426$. \\
\noindent
See \cite{jrfm17120531} for more details and methodologies on Fitting GTS distribution.

\section{Pricing European Options under GTS Process}
\subsection{GTS Process:  Risk Neutral Esscher Measure}
\noindent
The method of Esscher transform was introduced by \cite{gerber1995option}  as an efficient technique for
pricing derivative securities if a L\'evy process models the logarithms of the underlying asset prices. An Esscher transform of a stock-price process provides an equivalent martingale measure. The existence of the equivalent Esscher transform measure is not always guaranteed, and the issue of the unicity of the equivalent martingale measure remains recurrent when pricing an option.\\
\noindent
From the characteristic function $\vartheta(\xi,t)$ in (\ref{eq:l9}), the Moment generating function $M(h,t)$ of the Generalized Tempered Stable (GTS) distribution can be written as follows.
\begin{equation}
\begin{aligned}
 M(h,t)= \vartheta(-ih,t)=e^{t\Psi(-ih)}   \quad \hbox{with $ -\lambda_{-} < h < \lambda_{+} $} \label {eq:l11}
\end{aligned}
\end{equation}
Under the Esscher transform with parameter h, the probability density of $Y= \{Y_{t}\}$ becomes:
\begin{equation}
\begin{aligned}
\hat{f}(x,t,h)= \frac{e^{hx}{f(x, t)}}{M(h,t)}
\end{aligned}
\end{equation}
\noindent
The Moment generating function  $M(h,t)$ becomes M(z,t,h)
\begin{align}
M(z,t,h)= E^{h}\left[e^{zX_{t}}\right]=\int_{0}^{+\infty}e^{zx}\hat{f}(x,t,h)dx=\left(\frac{M(h+z,1)}{M(h,1)}\right)^{t} =M(z,1,h)^{t} \label {eq:l12}\\
M(z,1,h)= e^{\Psi_{h}(z)} \quad \quad \Psi_{h}(z)=\Psi(-i(h+z)) - \Psi(-i(h))) \label {eq:l13}
\end{align}

\noindent
Given the process, $\left\{e^{-r\tau}S(\tau)\right\}_{\tau \geq 0}$ with $r$ the constant risk-free rate of interest, we look into the conditions to have   $h^*$ such as
\begin{align}
E^{h^*}\left[e^{-r\tau}S(\tau)\right]=S(0) \label {eq:l14}
\end{align}
\noindent
$S(\tau)=S(0)e^{Y_{\tau}}$,  with $Y_{\tau}$ is the GTS process described by the characteristic exponent (\ref{eq:l7}). The equation (\ref {eq:l14}) becomes
\begin{equation}
\begin{aligned}
e^{r\tau}=E^{h^*}\left[e^{Y_{\tau}}\right]=M(1,1, h^*)^{\tau}=e^{\tau \Psi_{h^*}(1)} \quad \quad
\Psi_{h^*}(1)=r   \quad  \hbox{ with $ -\lambda_{-} < h^{*} < \lambda_{+} $ } \label {eq:l15}
\end{aligned}
\end{equation}
\noindent
The existence and unicity of $h^*$ in (\ref {eq:l15}) were studied empirically using the parameter data from Table \ref{tab1}. Over the interval $]-\lambda_{-} ; \lambda_{+}[$, $\Psi_{h}(1)$  is a strictly increasing function, as shown in Fig \ref{fig01a}. Fig \ref{fig01b} provides the solution $h^*$ of equation (\ref{eq:l15}) for free interest rate less than 10\%. The solution $h^*$ increases with free interest rate $r$.\\
\begin{figure}[ht]
\vspace{-0.5cm}
    \centering
  \begin{subfigure}[b]{0.45\linewidth}
    \includegraphics[width=\linewidth]{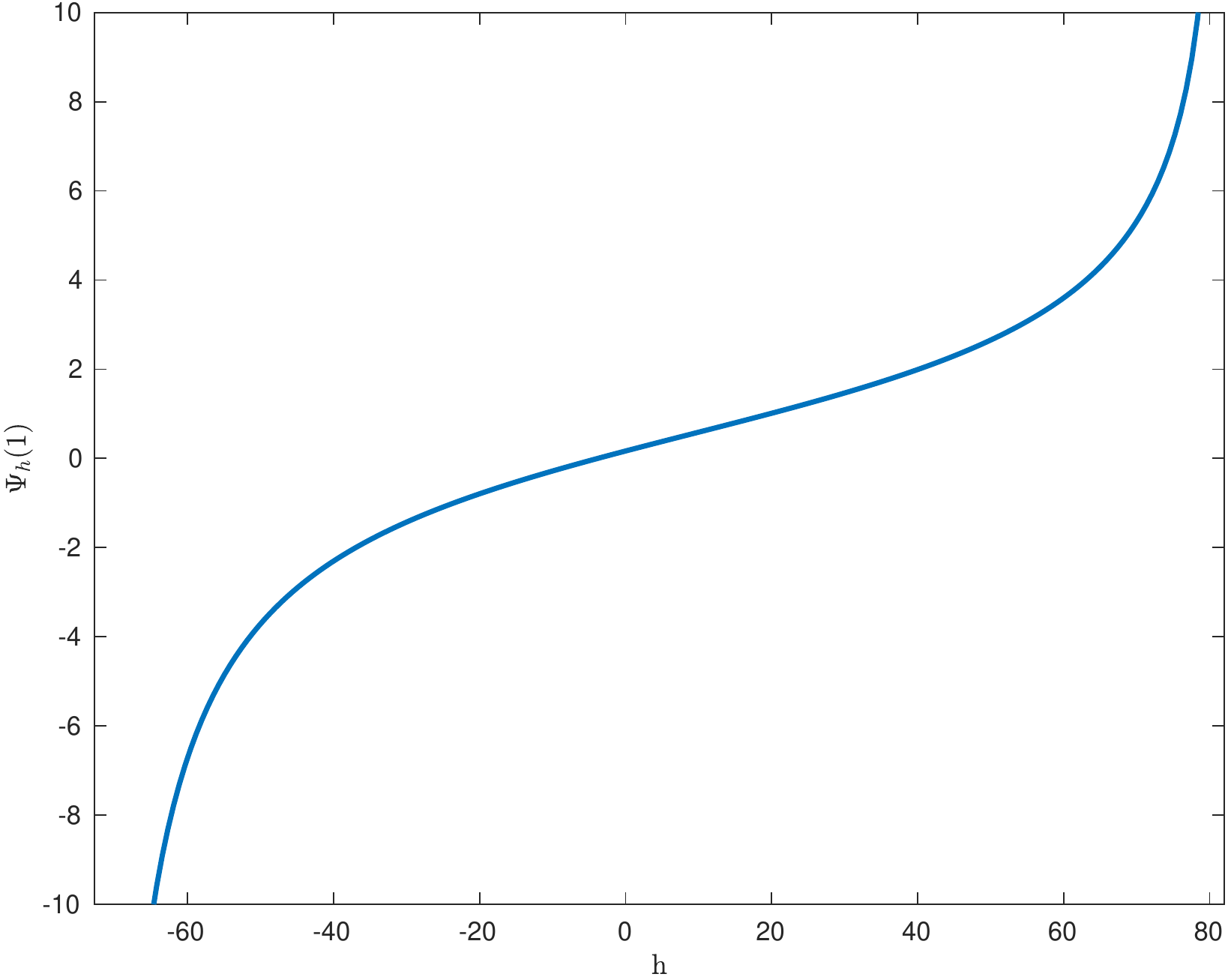}
\vspace{-0.5cm}
     \caption{$\Psi_{h}(1)$}
         \label{fig01a}
  \end{subfigure}
%\hspace{-0.3cm}
  \begin{subfigure}[b]{0.45\linewidth}
    \includegraphics[width=\linewidth]{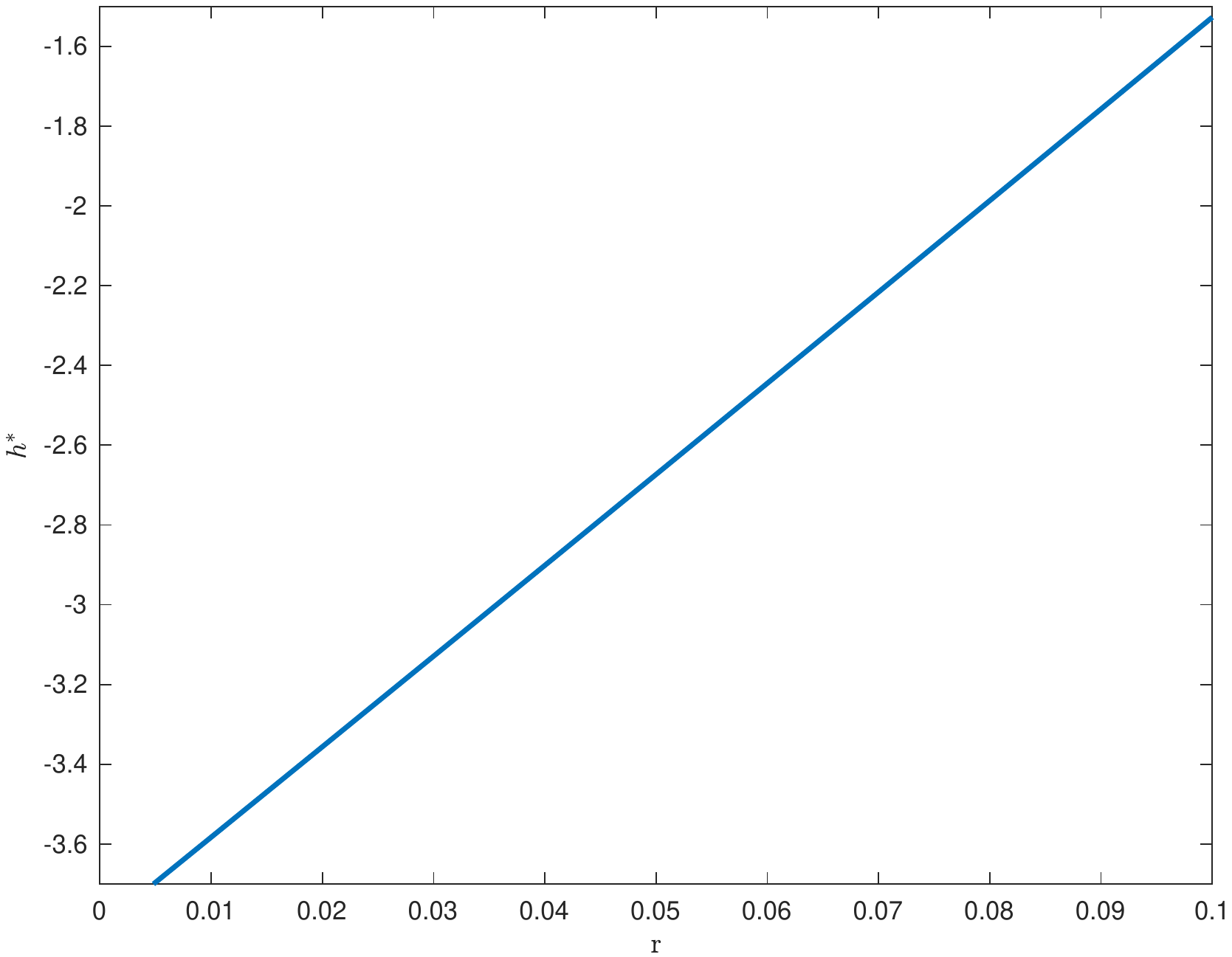}
\vspace{-0.5cm}
     \caption{$\Psi_{h^{*}}(1)=r$ }
         \label{fig01b}
          \end{subfigure}
\vspace{-0.7cm}
  \caption{$\mu=-0.693477$, $\beta_{+}=0.682290$, $\beta_{-}=0.242579$, $\alpha_{+}=0.458582$, $\alpha_{-}=0.414443$, $\lambda_{+}=0.822222$, $\lambda_{-}=0.727607$}
  \label{fig01}
\vspace{-0.7cm}
\end{figure}
%\newpage
 \begin{theorem}\label{lem71} (GTS Esscher transform distribution)\\
\noindent
Esscher transform of GTS process $Y=\{Y_{t}\}_{t \geq0} $  with parameter $(t\mu, \beta_{+}, \beta_{-}, t\alpha_{+},\alpha_{-}, \lambda_{+}, \lambda_{-})$ is also a is also a GTS process $Y^{*}=\{Y^{*}_{t}\}_{t \geq0} $  with parameter $(t\mu, \beta_{+}, \beta_{-}, t\alpha_{+}, t\alpha_{-}, \tilde{\lambda}_{+}, \tilde{\lambda}_{-})$ 
 \begin{align}
\tilde{\lambda}_{+}=\lambda_{+} - h^{*}  \quad \quad
\tilde{\lambda}_{-}=\lambda_{-} + h^{*} \label {eq:l16}\end{align}
\end{theorem}

\begin{proof} \ \\
we recall the function $\Psi_{h^{*} }(z)$ in (\ref{eq:l13}).
\noindent
\begin{equation}
\begin{aligned}
\Psi_{h^{*}}(z)&=\Psi(-i(h^{*} +z)) - \Psi(-i(h^{*} )))\\
&= \mu z + \alpha_{+}\Gamma(-\beta_{+})\left(((\lambda_{+} - h^{*} ) - z)^{\beta_{+}} - {(\lambda_{+} - h^{*} )}^{\beta_{+}}\right) + \\ &\alpha_{-}\Gamma(-\beta_{-})\left(((\lambda_{-} + h^{*} ) + z)^{\beta_{-}} - {(\lambda_{-} + h^{*} )}^{\beta_{-}}\right) \label {eq:l17}
\end{aligned}
\end{equation}
Using the Esscher transform  method, the moment generating function for GST process $Y=\{Y_{t}\}_{t \geq 0} $ becomes:
\begin{align}
M(z,t,{h^{*}})= E^{{h^{*}}}\left[e^{zY_{t}}\right]=M(z,1,{h^{*}})^{t} = e^{t\Psi_{h^{*}}(z)} \quad \hbox{with $ -\lambda_{-} - h^{*} \leq z \leq \lambda_{+} - h^{*}$ } \label {eq:l18}
\end{align}
\noindent
The Esscher transform method preserves the structure of the GTS process. We have a new GST process generated by the GST distribution $GST(\mu, \beta_{+}, \beta_{-}, \alpha_{+},\alpha_{-}, \lambda_{+} - h^{*}, \lambda_{-} + h^{*})$. 
\end{proof}
%Option Pricing under Variance Gamma model: Extended Black-Scholes Formula (Framework)%\clearpage
%Pricing European Options under the 5 Parameters SPY Variance-Gamma Model: Extended Black-Scholes Formula
%\clearpage
%\newpage

\subsection{Extended Black-Scholes Formula}
\noindent
Under the EMM, $f(\xi, \tau, h^{*})$ is the probability density of  GST distribution with parameter \\ $(\tau\mu, \beta_{+}, \beta_{-}, \tau\alpha_{+}, \tau\alpha_{-}, \tilde{\lambda}_{+}, \tilde{\lambda}_{-})$. ($ \tilde{\lambda}_{+}, \tilde{\lambda}_{-}$) is defined in (\ref{eq:l16}).
\begin{corollary}\label{lem8}
(Extended Black-Scholes)\ \\
  Let $r$ a continuously compounded risk-free rate of interest; $Y=\{Y_{t}\}_{t \geq0} $, a GST Process with parameter $(t\mu, \beta_{+}, \beta_{-}, t\alpha_{+}, t\alpha_{-}, \tilde{\lambda}_{+}, \tilde{\lambda}_{-})$; and $(S(0)e^{X_{T}} - K)^{+}$ the terminal payoff for a contingent claim with the expiry date $T$.
Then at time $t <T$, the arbitrage price of the European call option with the strike price $K$ can be written as follows.
 \begin{align}
F^{GTS}_{call}(S_{t},t)&= S(t)\left[ 1- F(log(\frac{K}{S(t)}),\tau, h^{*}+1) \right] - Ke^{-r\tau}\left[ 1- F(log(\frac{K}{S(t)}),\tau, h^{*})\right] \label {eq:l15}\\
 &F(log(\frac{K}{S(t)}),\tau, h^{*})=\int_{-\infty }^{log(\frac{K}{S(t)})}f(\xi, \tau, h^{*})d\xi \label {eq:l19}
\end{align}

Where $\tau=T-t$, $F(k,\tau, h^{*})$ and $F(k,\tau, h^{*}+1)$ are the cumulative distribution of GST distribution. 
\end{corollary}
\noindent
See  \cite{nzokem2022} for corollary \ref{lem8} proof \\

\noindent
The GTS distribution function's probability density and cumulative functions were obtained through the inverse Fourier transform. See \cite{nzokem2021fitting,nzokem_2021b,nzokem2022} for details.
 \begin{align}
 f(\xi, \tau, h^{*})=\frac{1}{2\pi}\int_{-\infty}^{+\infty}e^{i \xi z + \tau \Psi_{h^{*}}(-z)}dz \quad  \quad F(\xi, \tau, h^{*})=\frac{1}{2\pi}\int_{-\infty}^{+\infty}\frac{e^{i \xi z + \tau \Psi_{h^{*}}(-z)}}{iz}dz +\frac{1}{2} \label {eq:l20a} 
\end{align}
\begin{figure}[ht]
\vspace{-0.3cm}
    \centering
\hspace{-1cm}
  \begin{subfigure}[b]{0.33\linewidth}
    \includegraphics[width=\linewidth]{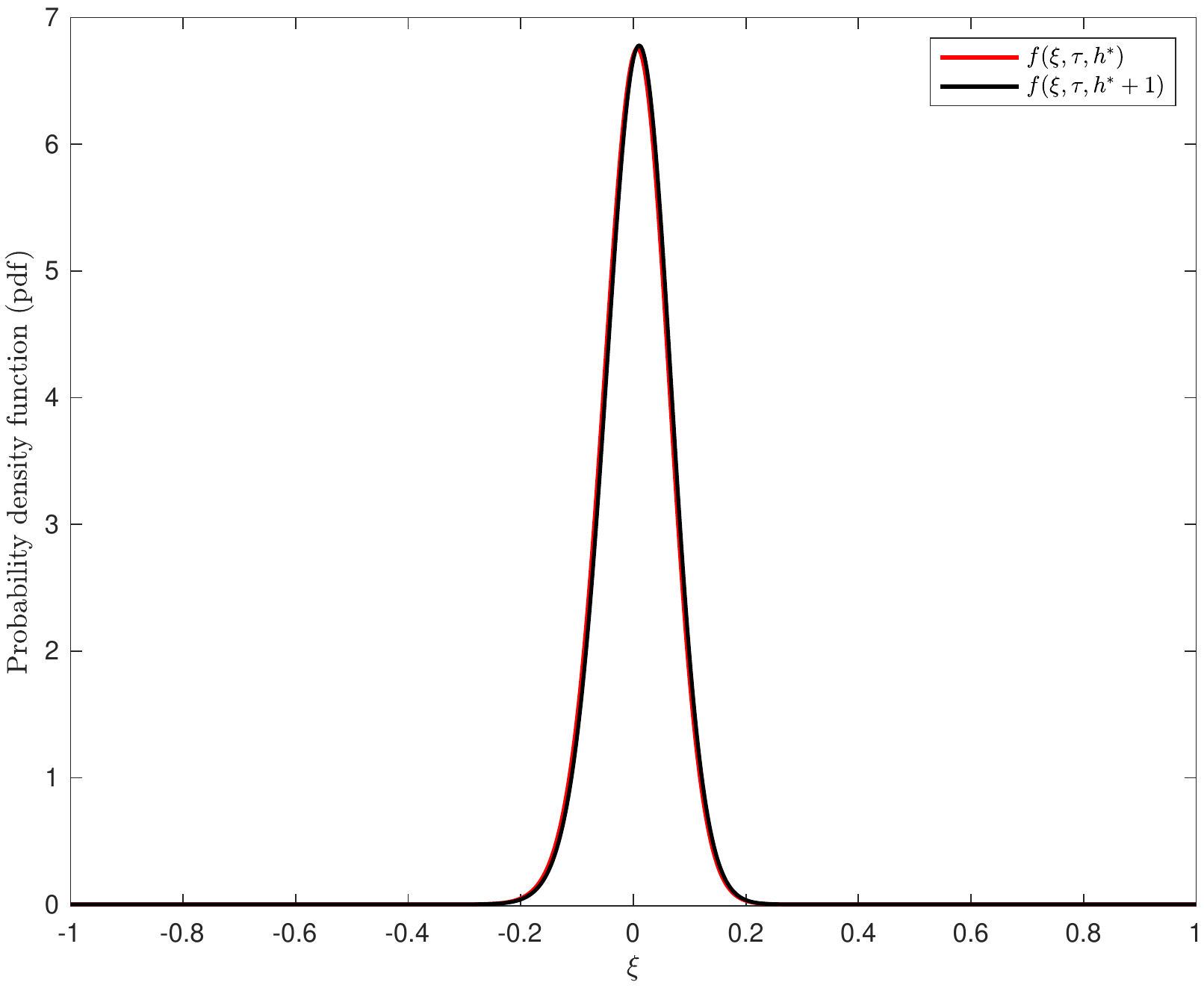}
\vspace{-0.5cm}
     \caption{$\tau=\frac{1}{12}$ years}
         \label{fig71}
  \end{subfigure}
\hspace{-0.2cm}
  \begin{subfigure}[b]{0.33\linewidth}
    \includegraphics[width=\linewidth]{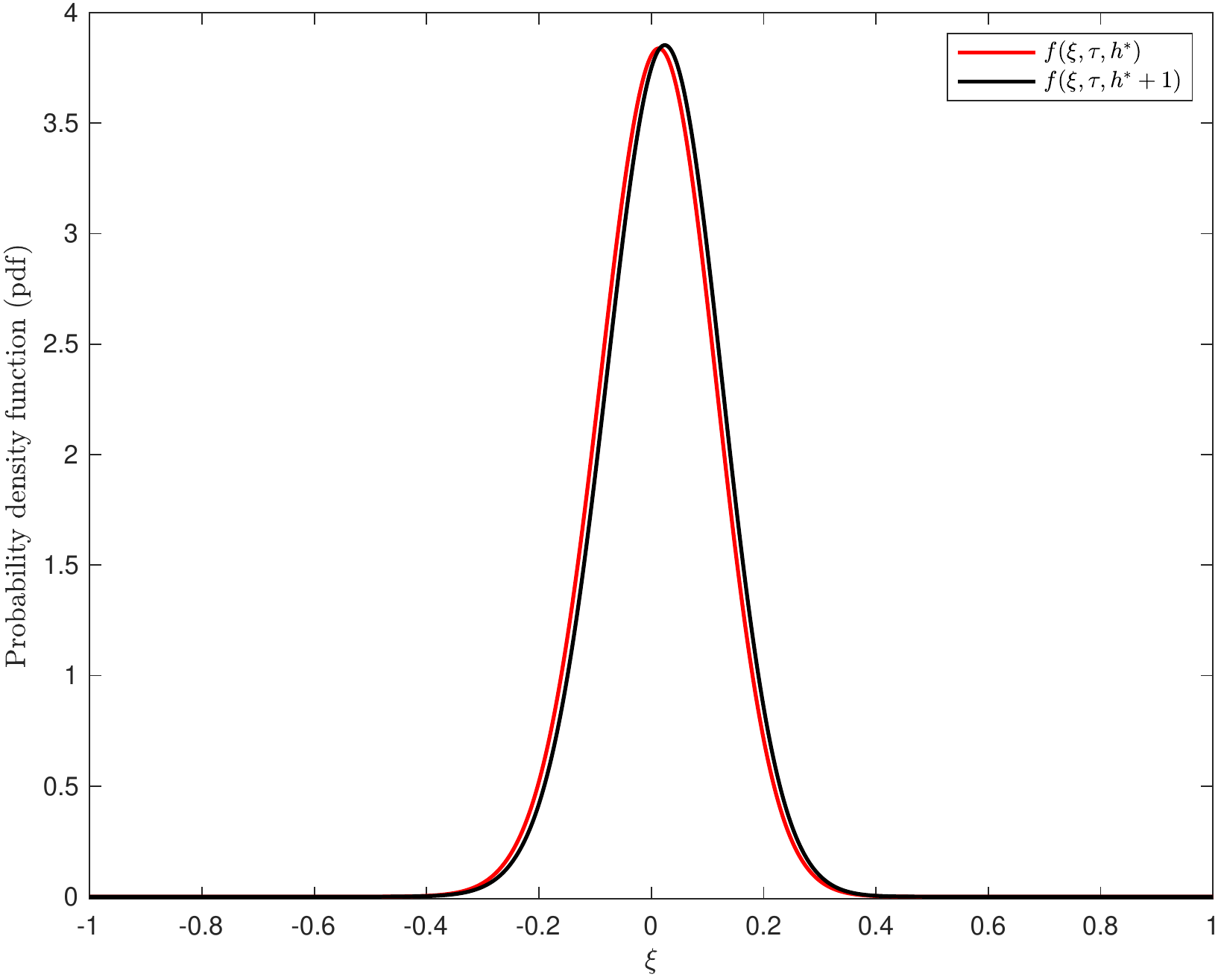}
\vspace{-0.5cm}
     \caption{$\tau=0.25$ years}
         \label{fig72}
          \end{subfigure}
\hspace{-0.2cm}
  \begin{subfigure}[b]{0.33\linewidth}
    \includegraphics[width=\linewidth]{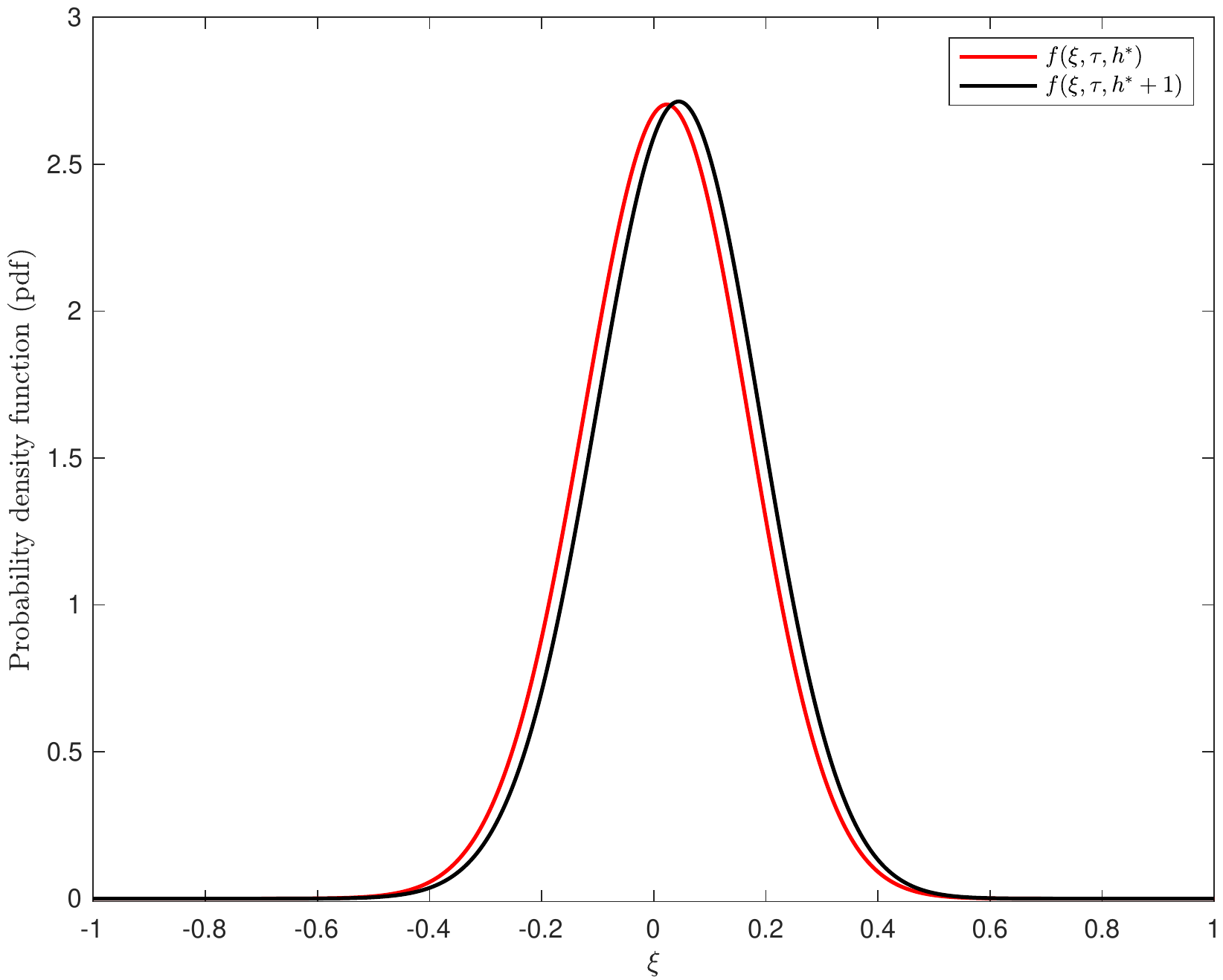}
\vspace{-0.5cm}
     \caption{$\tau=0.5$ years}
         \label{fig73}
          \end{subfigure}
\vspace{-0.5cm}
  \caption{Estimations: $f(\xi,\tau, h^{*})$ versus $f(\xi, \tau, h^{*} + 1)$}
  \label{fig12}
\vspace{-0.7cm}
\end{figure}

\noindent
Based on parameter data from Table \ref{tab1}, we added a 6\% risk-free interest rate and the corresponding Esscher transform parameter ($h^{*}=-2.4448$) in Fig \ref{fig01b}. FRFT technique was used to compute the density and cumulative functions in (\ref{eq:l20a}).\\
\begin{figure}[ht]
\vspace{-0.5cm}
    \centering
\hspace{-1cm}
  \begin{subfigure}[b]{0.33\linewidth}
    \includegraphics[width=\linewidth]{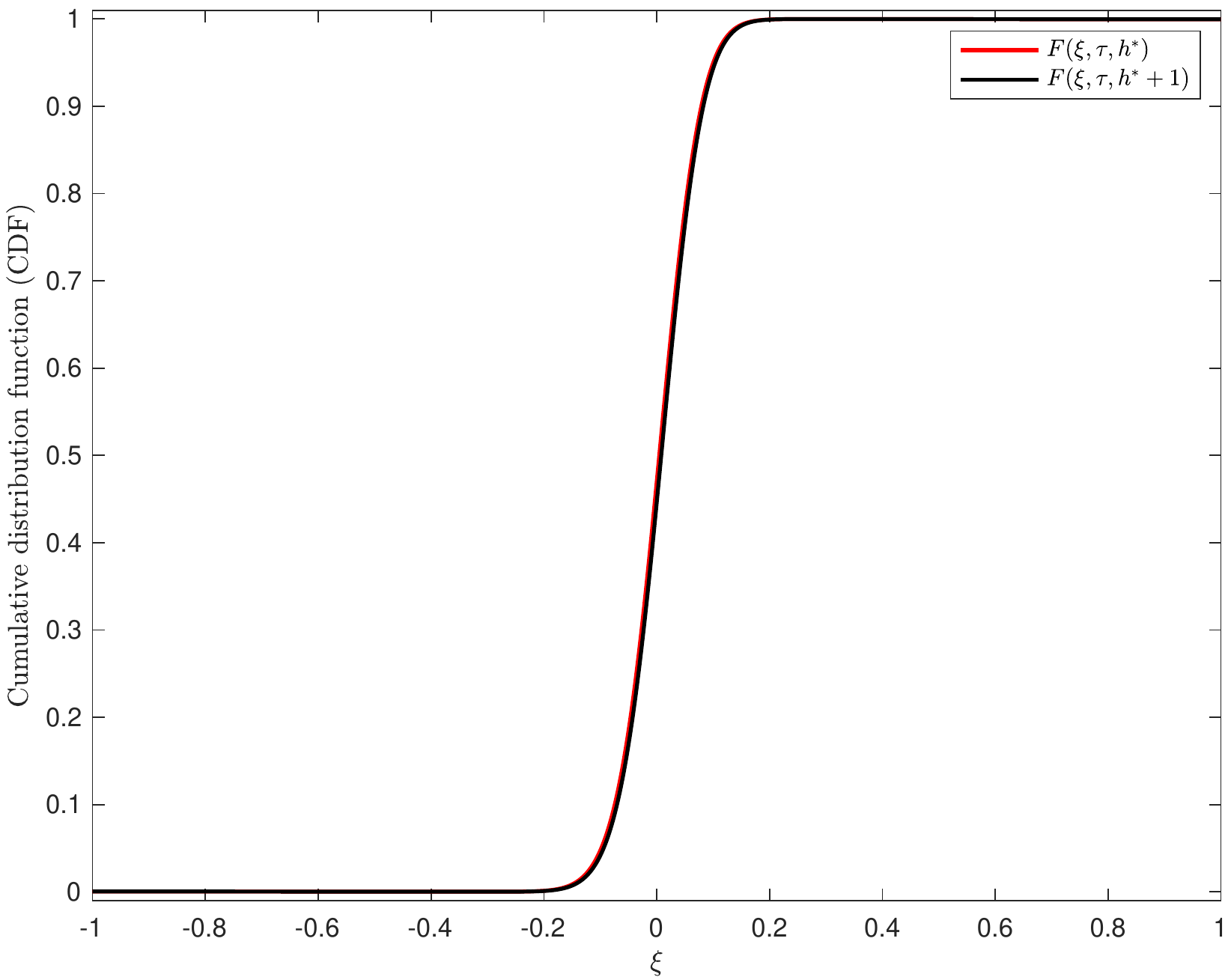}
\vspace{-0.5cm}
     \caption{$\tau=\frac{1}{12}$ years}
         \label{fig71}
  \end{subfigure}
\hspace{-0.2cm}
  \begin{subfigure}[b]{0.33\linewidth}
    \includegraphics[width=\linewidth]{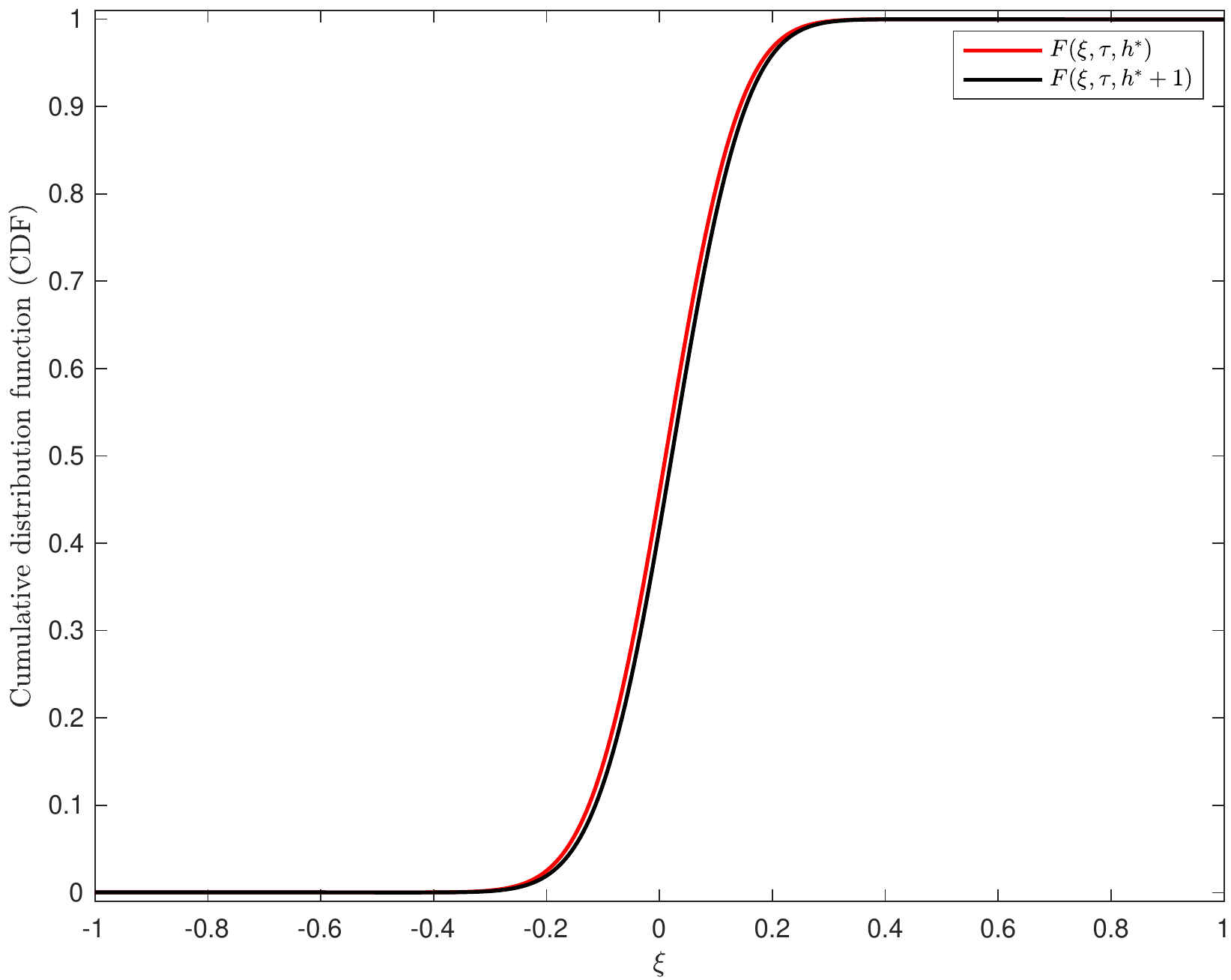}
\vspace{-0.5cm}
     \caption{$\tau=0.25$ years}
         \label{fig72}
          \end{subfigure}
\hspace{-0.2cm}
  \begin{subfigure}[b]{0.33\linewidth}
    \includegraphics[width=\linewidth]{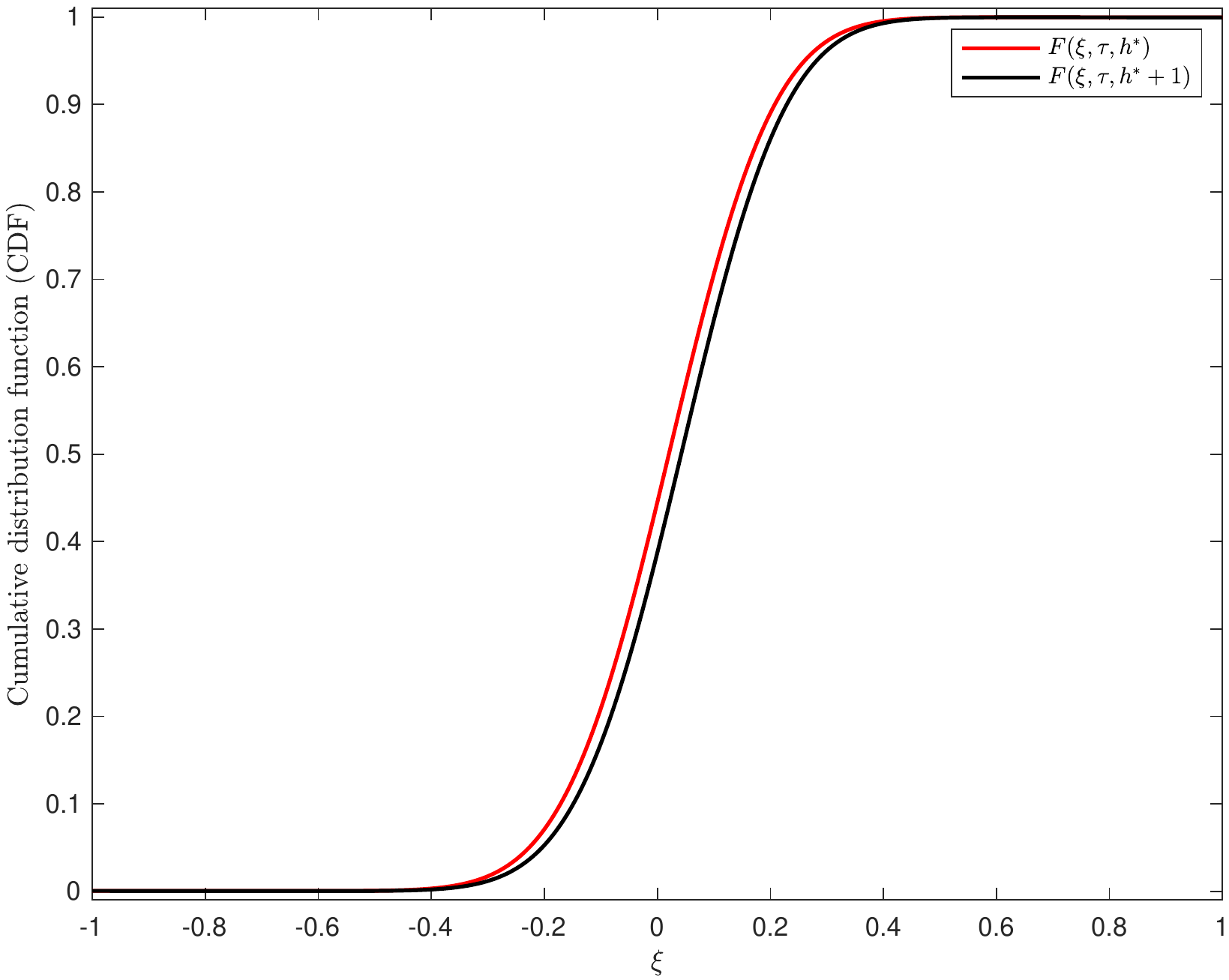}
\vspace{-0.5cm}
     \caption{$\tau=0.5$ years}
         \label{fig73}
          \end{subfigure}
\vspace{-0.5cm}
  \caption{Estimations: $F(\xi,\tau, h^{*})$ versus $F(\xi, \tau, h^{*} + 1)$}
  \label{fig13}
\vspace{-0.5cm}
\end{figure}

\noindent
As shown in Fig \ref{fig12} and \ref{fig13}, the discrepancy between both functions increases with the period ($\tau$). The estimation results will be used to compute the option pricing (\ref{eq:l19}).

%\newpage
\subsection{ Generalized Black-Scholes Formula}
\noindent
Under the EMM, $f(\xi, \tau, h^{*})$ can be written as follows:
 \begin{align}
f(\xi, \tau, h^{*})=\frac{1}{2\pi}\int_{-\infty}^{+\infty}e^{-i \xi z - \tau \varphi(z)}dz \quad \quad   
\varphi(z)=-\Psi_{h^{*}}(z) \label {eq:l20} 
\end{align}
with $\Psi_{h^{*}}(z)$ defined in (\ref{eq:l12}).

 \begin{theorem} \label{lem9}\ \\
 \noindent
 Let $r$ a continuously compounded risk-free rate of interest; $Y=\{Y_{t}\}_{t \geq0} $, a Variance-Gamma Process with parameter $(\mu t, \delta, \sigma, \alpha t, \theta)$; and $(S(0)e^{X_{T}} - K)^{+}$, the terminal payoff for a contingent claim with the expiry date $T$.
Then at time $t <T$, the arbitrage price of the European call option with the strike price $K$ can be written as follows.
\begin{align}
F^{GTS}_{call}(S_{t},t)=\frac{K}{2\pi} \int_{-\infty + i q}^{+\infty + i q}{\frac{e^{\left(i \xi log(\frac{S(t)}{K}) - \tau (r + \varphi(\xi))\right)}}{i \xi(i\xi -1)} d\xi} \label{eq:l21}
\end{align}
Where $\varphi(z)$ is the characteristic exponent of the VG model, $\tau=T-t$, and $q < -1$.
\end{theorem}
\noindent
See  \cite{nzokem2022} for theorem \ref{lem9} proof \\

 \noindent
To compute $F^{GTS}_{call}$ in (\ref{eq:l21}), the numerical integration technique, called the  Newton-Cotes formula of 12 degrees, will be implemented as follows.
\begin{align}
g^{x}(\xi)= \frac{exp\left[{i(\xi+iq)x-\tau (r -\Psi_{h^{*}}(\xi+iq))}\right]}{2\pi i(\xi+iq)(i(\xi+iq)-1)}  \quad 
\widehat{F^{GTS}_{call}}=\frac{b}{n}\sum_{p=0}^{\frac{n}{Q} - 1} \sum_{j=0}^{Q} W_{j} g^{log(\frac{S(t)}{K})}(\xi_{Qp+j}) \label {eq:l22} 
\end{align}

 \noindent
 In order to perform the integral (\ref{eq:l22}), the following parameter values were used: $a=0$, $b=20$, $Q=12$, $n=5000Q$, $n_{0}=5000$. See \cite{aubain2020, nzokem_2021} for more details on the choice of parameter values and the twelve-point rule Composite Newton--Cote's methodology.
 
 \subsection{Parameter $q$ Evaluation}
\noindent
Let us consider the stock or index price $S=S_{0} e^{Y}$ and the strike price $K$; the call payoff can be written as follows.
  \begin{align}
 (S(0)e^{Y_{T}} - K)^{+}&= S(t)(e^{Y_{\tau}} - k)^{+}=S(t)g(Y_{\tau},k) \   \  \   \   k=\frac{K}{S_{t}} \label {eq:l23}  
\end{align}
The Fourier transform of the call payoff  in (\ref{eq:l23}) can be derived as follows
  \begin{align}
  \scrF[g] (y,k)&=\int_{0}^{+\infty}e^{-iyx}g(x,k)dx= \frac{ke^{-iylog(k)}}{iy(iy -1)} \quad  \hbox{for $\Im(y) < -1$} 
\end{align}
\noindent
Call payoff in (\ref{eq:l23})  can be recovered from the inverse of Fourier and labeled $\check{g}(x,k)$ with 
  \begin{align}
  \check{g}(x,k)=\frac{1}{2\pi}\int_{-\infty + iq}^{+\infty + iq}e^{iyx}\scrF[g] (y,k)dy \hspace{5mm}  \hbox{for $q < -1$} 
 \label {eq:l24}
\end{align}

The payoff ($\check{g}(x,k)$) in (\ref{eq:l24}) depends on the $q$ parameter as shown in Fig \ref{fig66}. The inverse Fourier ($\check{g}(x,k)$) produces poor results if $q=-2$ (color red).\\

\noindent
To find the $q$ value with a high level of accuracy, we define the error function ($ER(k,q)$) between the call payoff in (\ref{eq:l23}) and the inverse Fourier payoff in (\ref{eq:l24}).

  \begin{align}
ER(k,q)=\sqrt{ \frac{1}{m}\sum_{j=1}^{m} \left[ ( e^{x_{j}} - k)^{+} -  \check{g}(x_{j},k) \right]^2}  \hspace{5mm}  \hbox{with $-M \leq x_{j}\leq M$} 
 \label {eq:l25}
\end{align}

\begin{figure}[ht]
\vspace{-0.6cm}
    \centering
\hspace{-0.3cm}
  \begin{subfigure}[b]{0.42\linewidth}
    \includegraphics[width=\linewidth]{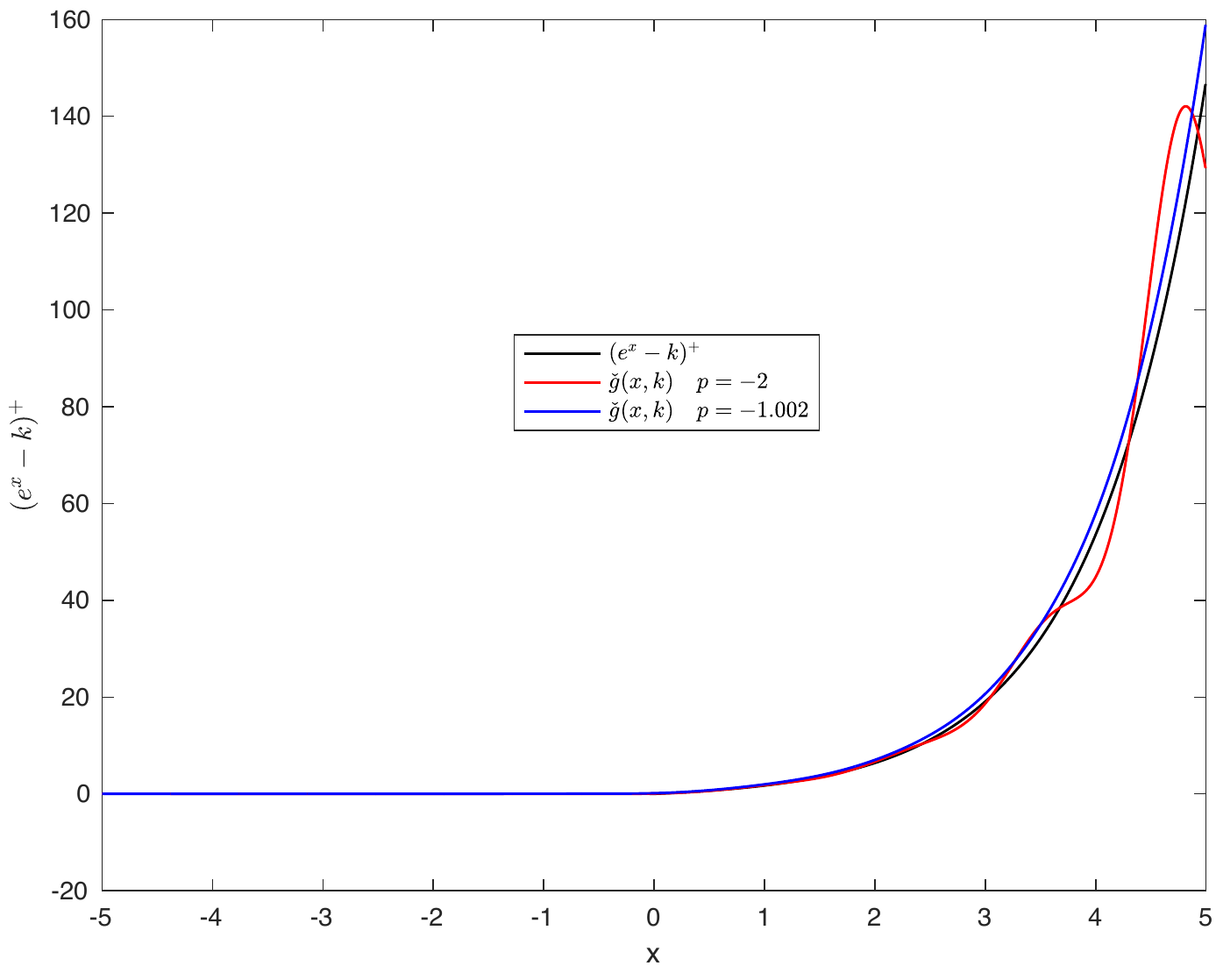}
\vspace{-0.5cm}
     \caption{$(X - k)^{+}$ versus $\check{g}(x,k)$}
         \label{fig66}
  \end{subfigure}
\hspace{-0.4cm}
  \begin{subfigure}[b]{0.46 \linewidth}
    \includegraphics[width=\linewidth]{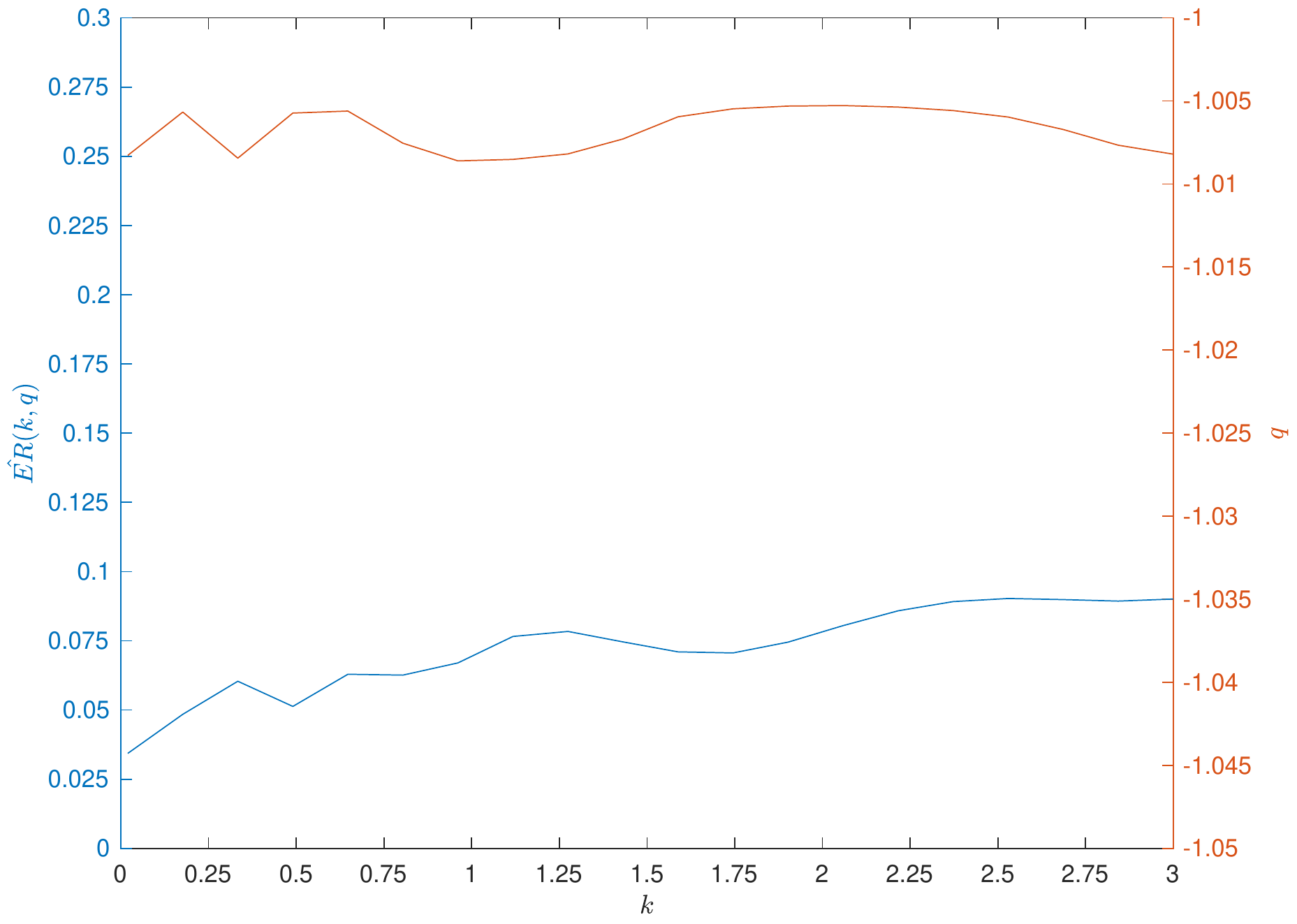}
\vspace{-0.5cm}
     \caption{$ER(k,q)$ and Optimal parameter (q)}
         \label{fig70}
          \end{subfigure}
\vspace{-0.5cm}
  \caption{Optimal parameter ($q$) and minimum Error value ($ER(k,q)$)}
  \label{fig66a}
\vspace{-0.6cm}
\end{figure}

\noindent
Fig \ref{fig70} displays (in color blue) the $ER(k,q)$ minimum value as a function of the strike price $k$; and the optimal correspondent parameter $q$ as a function of the strike price $k$ (in color red). Both graphs display almost a constant function with respect to the strike price.
\section{Empirical Analysis: A case of S\&P 500 return}
\noindent
Based on parameter data from Table \ref{tab1},  we added a 6\% risk-free interest rate and the corresponding Esscher transform parameter ($h^{*}=-2.4448$) in Fig \ref{fig01b}. The  GTS  option pricing will be computed across maturity and option moneyness using Extended and Generalized Black-Scholes formulas. The closed-form Black-Scholes model \cite{hull2003options} was added to the analysis as a benchmark.
\begin{align}
F^{BS}_{call}(S_{t},\tau)&=S_{t}N(d_{1}) - K e^{-r\tau}N(d_{2}) \label {eq:l25}\\
d_{1}=\frac{Ln(\frac{S_{t}}{K}) + (r+\frac{1}{2}\sigma^{*2})\tau}{\sigma^{*} \sqrt{\tau}}
 \quad \quad d_{2}&=d_{1} -\sigma^{*}\sqrt{\tau} \quad \quad N(x)=\frac{1}{\sqrt{2\pi}}\int_{-\infty}^x\exp\left( -\frac{t^2}{2}\ \right)\,dt  \label {eq:l26}
\end{align}

\noindent
The variance $\sigma^{*}=0.2077$  is the annualized variance computed from the variance in Table \ref{tab2}.\\
 Moneyness describes the intrinsic value of an option in its current state, which indicates whether the option would make money if exercised immediately. Option moneyness can be classified into three categories: At-the-money (ATM) option ($k=\frac{S_{t}}{K}=1$), Out-of-the-money (OTM) option ($k=\frac{S_{t}}{K}<1$), and In-the-money (ITM) option ($k=\frac{S_{t}}{K}>1$). \\
 \noindent
 On August 15, 2023, the S\&P 500 market price closed at $\$4,437.86$. We compute the GTS call option price on S\&P 500 with the spot price ($S_{0}$) $\$4,437.86$.  Table \ref{tab3} summarizes the computation of price as a function of time to maturity ($\tau$) and option moneyness ($k=\frac{S_{t}}{K}$).
\begin{table}[ht]
\centering
\vspace{-0.6cm}
%\centering
\caption{European Call Price on S\&P 500 data}
%\vspace{-0.3cm}
\label{tab3}
\setlength\extrarowheight{0.5pt} % for a less cramped look
\setlength\tabcolsep{0.5pt}  
%\resizebox{\columnwidth}{!}{%
%\resizebox{15cm}{!}{%
\begin{tabular}{@{}c|c|ccc|ccc|ccc|ccc@{}}
\toprule
\multirow{1}{*}{\textbf{Strike Price}} & \multirow{1}{*}{\textbf{$\frac{S_{t}}{K}$}} & \multirow{1}{*}{\textbf{BSM}} & \multirow{1}{*}{\textbf{GTS(\ref{eq:l21})}}& \multirow{1}{*}{\textbf{GTS(\ref{eq:l16})}} & \multirow{1}{*}{\textbf{BSM}} & \multirow{1}{*}{\textbf{GTS(\ref{eq:l21})}} & \multirow{1}{*}{\textbf{GTS(\ref{eq:l16})}} & \multirow{1}{*}{\textbf{BSM}} & \multirow{1}{*}{\textbf{GTS(\ref{eq:l21})}} & \multirow{1}{*}{\textbf{GTS(\ref{eq:l16})}} & \multirow{1}{*}{\textbf{BSM}} & \multirow{1}{*}{\textbf{GTS(\ref{eq:l21})}} & \multirow{1}{*}{\textbf{GTS(\ref{eq:l16})}} \\ \midrule
\multicolumn{2}{c|}{\textbf{Period ( in yrs)}} & \multicolumn{3}{c|}{\textbf{0.25}} & \multicolumn{3}{c|}{\textbf{0.5}} & \multicolumn{3}{c|}{\textbf{0.75}} & \multicolumn{3}{c|}{\textbf{1}} \\ \midrule
\multirow{1}{*}{2689.61} & \multirow{1}{*}{1.65} & \multirow{1}{*}{1788.29} & \multirow{1}{*}{1788.30} & \multirow{1}{*}{1788.29} & \multirow{1}{*}{1827.76} & \multirow{1}{*}{1827.78} & \multirow{1}{*}{1827.78} & \multirow{1}{*}{1866.80} & \multirow{1}{*}{1866.91} & \multirow{1}{*}{1866.90} & \multirow{1}{*}{1905.61} & \multirow{1}{*}{1905.83} & \multirow{1}{*}{1905.83} \\ \midrule
\multirow{1}{*}{2773.66} & \multirow{1}{*}{1.60} & \multirow{1}{*}{1705.49} & \multirow{1}{*}{1705.50} & \multirow{1}{*}{1705.49} & \multirow{1}{*}{1746.22} & \multirow{1}{*}{1746.26} & \multirow{1}{*}{1746.25} & \multirow{1}{*}{1786.62} & \multirow{1}{*}{1786.78} & \multirow{1}{*}{1786.77} & \multirow{1}{*}{1826.92} & \multirow{1}{*}{1827.23} & \multirow{1}{*}{1827.22} \\ \midrule
\multirow{1}{*}{2863.14} & \multirow{1}{*}{1.55} & \multirow{1}{*}{1617.35} & \multirow{1}{*}{1617.36} & \multirow{1}{*}{1617.35} & \multirow{1}{*}{1659.44} & \multirow{1}{*}{1659.52} & \multirow{1}{*}{1659.51} & \multirow{1}{*}{1701.40} & \multirow{1}{*}{1701.63} & \multirow{1}{*}{1701.63} & \multirow{1}{*}{1743.41} & \multirow{1}{*}{1743.83} & \multirow{1}{*}{1743.83} \\ \midrule
\multirow{1}{*}{2958.57} & \multirow{1}{*}{1.50} & \multirow{1}{*}{1523.34} & \multirow{1}{*}{1523.35} & \multirow{1}{*}{1523.34} & \multirow{1}{*}{1566.95} & \multirow{1}{*}{1567.08} & \multirow{1}{*}{1567.08} & \multirow{1}{*}{1610.73} & \multirow{1}{*}{1611.07} & \multirow{1}{*}{1611.07} & \multirow{1}{*}{1654.74} & \multirow{1}{*}{1655.31} & \multirow{1}{*}{1655.31} \\ \midrule
\multirow{1}{*}{3060.59} & \multirow{1}{*}{1.45} & \multirow{1}{*}{1422.84} & \multirow{1}{*}{1422.87} & \multirow{1}{*}{1422.86} & \multirow{1}{*}{1468.23} & \multirow{1}{*}{1468.44} & \multirow{1}{*}{1468.44} & \multirow{1}{*}{1514.20} & \multirow{1}{*}{1514.70} & \multirow{1}{*}{1514.70} & \multirow{1}{*}{1560.59} & \multirow{1}{*}{1561.35} & \multirow{1}{*}{1561.35} \\ \midrule
\multirow{1}{*}{3169.90} & \multirow{1}{*}{1.40} & \multirow{1}{*}{1315.19} & \multirow{1}{*}{1315.24} & \multirow{1}{*}{1315.24} & \multirow{1}{*}{1362.75} & \multirow{1}{*}{1363.11} & \multirow{1}{*}{1363.11} & \multirow{1}{*}{1411.46} & \multirow{1}{*}{1412.17} & \multirow{1}{*}{1412.17} & \multirow{1}{*}{1460.70} & \multirow{1}{*}{1461.70} & \multirow{1}{*}{1461.70} \\ \midrule
\multirow{1}{*}{3287.30} & \multirow{1}{*}{1.35} & \multirow{1}{*}{1199.63} & \multirow{1}{*}{1199.76} & \multirow{1}{*}{1199.75} & \multirow{1}{*}{1250.06} & \multirow{1}{*}{1250.64} & \multirow{1}{*}{1250.64} & \multirow{1}{*}{1302.26} & \multirow{1}{*}{1303.25} & \multirow{1}{*}{1303.24} & \multirow{1}{*}{1354.91} & \multirow{1}{*}{1356.20} & \multirow{1}{*}{1356.20} \\ \midrule
\multirow{1}{*}{3413.74} & \multirow{1}{*}{1.30} & \multirow{1}{*}{1075.41} & \multirow{1}{*}{1075.69} & \multirow{1}{*}{1075.69} & \multirow{1}{*}{1129.91} & \multirow{1}{*}{1130.79} & \multirow{1}{*}{1130.79} & \multirow{1}{*}{1186.56} & \multirow{1}{*}{1187.87} & \multirow{1}{*}{1187.87} & \multirow{1}{*}{1243.25} & \multirow{1}{*}{1244.87} & \multirow{1}{*}{1244.87} \\ \midrule
\multirow{1}{*}{3550.29} & \multirow{1}{*}{1.25} & \multirow{1}{*}{941.96} & \multirow{1}{*}{942.52} & \multirow{1}{*}{942.52} & \multirow{1}{*}{1002.39} & \multirow{1}{*}{1003.67} & \multirow{1}{*}{1003.66} & \multirow{1}{*}{1064.63} & \multirow{1}{*}{1066.33} & \multirow{1}{*}{1066.32} & \multirow{1}{*}{1126.01} & \multirow{1}{*}{1127.99} & \multirow{1}{*}{1127.99} \\ \midrule
\multirow{1}{*}{3698.22} & \multirow{1}{*}{1.20} & \multirow{1}{*}{799.32} & \multirow{1}{*}{800.34} & \multirow{1}{*}{800.34} & \multirow{1}{*}{868.33} & \multirow{1}{*}{870.04} & \multirow{1}{*}{870.03} & \multirow{1}{*}{937.30} & \multirow{1}{*}{939.38} & \multirow{1}{*}{939.38} & \multirow{1}{*}{1003.88} & \multirow{1}{*}{1006.21} & \multirow{1}{*}{1006.21} \\ \midrule
\multirow{1}{*}{3859.01} & \multirow{1}{*}{1.15} & \multirow{1}{*}{649.08} & \multirow{1}{*}{650.65} & \multirow{1}{*}{650.65} & \multirow{1}{*}{729.59} & \multirow{1}{*}{731.69} & \multirow{1}{*}{731.69} & \multirow{1}{*}{806.08} & \multirow{1}{*}{808.50} & \multirow{1}{*}{808.49} & \multirow{1}{*}{878.05} & \multirow{1}{*}{880.69} & \multirow{1}{*}{880.69} \\ \midrule
\multirow{1}{*}{4034.42} & \multirow{1}{*}{1.10} & \multirow{1}{*}{495.72} & \multirow{1}{*}{497.69} & \multirow{1}{*}{497.69} & \multirow{1}{*}{589.51} & \multirow{1}{*}{591.85} & \multirow{1}{*}{591.85} & \multirow{1}{*}{673.41} & \multirow{1}{*}{676.02} & \multirow{1}{*}{676.02} & \multirow{1}{*}{750.36} & \multirow{1}{*}{753.20} & \multirow{1}{*}{753.20} \\ \midrule
\multirow{1}{*}{4226.53} & \multirow{1}{*}{1.05} & \multirow{1}{*}{347.77} & \multirow{1}{*}{349.64} & \multirow{1}{*}{349.63} & \multirow{1}{*}{453.12} & \multirow{1}{*}{455.39} & \multirow{1}{*}{455.39} & \multirow{1}{*}{542.70} & \multirow{1}{*}{545.31} & \multirow{1}{*}{545.31} & \multirow{1}{*}{623.36} & \multirow{1}{*}{626.25} & \multirow{1}{*}{626.25} \\ \midrule
\multirow{1}{*}{4437.86} & \multirow{1}{*}{1.00} & \multirow{1}{*}{217.36} & \multirow{1}{*}{218.45} & \multirow{1}{*}{218.45} & \multirow{1}{*}{326.85} & \multirow{1}{*}{328.69} & \multirow{1}{*}{328.69} & \multirow{1}{*}{418.34} & \multirow{1}{*}{420.69} & \multirow{1}{*}{420.68} & \multirow{1}{*}{500.33} & \multirow{1}{*}{503.05} & \multirow{1}{*}{503.05} \\ \midrule
\multirow{1}{*}{4671.43} & \multirow{1}{*}{0.95} & \multirow{1}{*}{116.48} & \multirow{1}{*}{116.53} & \multirow{1}{*}{116.53} & \multirow{1}{*}{217.59} & \multirow{1}{*}{218.71} & \multirow{1}{*}{218.71} & \multirow{1}{*}{305.24} & \multirow{1}{*}{307.06} & \multirow{1}{*}{307.06} & \multirow{1}{*}{385.05} & \multirow{1}{*}{387.40} & \multirow{1}{*}{387.40} \\ \midrule
\multirow{1}{*}{4930.96} & \multirow{1}{*}{0.90} & \multirow{1}{*}{51.10} & \multirow{1}{*}{50.51} & \multirow{1}{*}{50.51} & \multirow{1}{*}{130.98} & \multirow{1}{*}{131.32} & \multirow{1}{*}{131.32} & \multirow{1}{*}{208.12} & \multirow{1}{*}{209.26} & \multirow{1}{*}{209.26} & \multirow{1}{*}{281.52} & \multirow{1}{*}{283.31} & \multirow{1}{*}{283.31} \\ \midrule
\multirow{1}{*}{5221.01} & \multirow{1}{*}{0.85} & \multirow{1}{*}{17.38} & \multirow{1}{*}{16.80} & \multirow{1}{*}{16.80} & \multirow{1}{*}{69.54} & \multirow{1}{*}{69.31} & \multirow{1}{*}{69.31} & \multirow{1}{*}{130.53} & \multirow{1}{*}{131.00} & \multirow{1}{*}{130.99} & \multirow{1}{*}{193.32} & \multirow{1}{*}{194.46} & \multirow{1}{*}{194.46} \\ \midrule
\multirow{1}{*}{5547.33} & \multirow{1}{*}{0.80} & \multirow{1}{*}{4.29} & \multirow{1}{*}{4.02} & \multirow{1}{*}{4.02} & \multirow{1}{*}{31.60} & \multirow{1}{*}{31.15} & \multirow{1}{*}{31.15} & \multirow{1}{*}{73.87} & \multirow{1}{*}{73.82} & \multirow{1}{*}{73.82} & \multirow{1}{*}{122.96} & \multirow{1}{*}{123.47} & \multirow{1}{*}{123.47} \\ \midrule
\multirow{1}{*}{5917.15} & \multirow{1}{*}{0.75} & \multirow{1}{*}{0.71} & \multirow{1}{*}{0.65} & \multirow{1}{*}{0.65} & \multirow{1}{*}{11.84} & \multirow{1}{*}{11.47} & \multirow{1}{*}{11.47} & \multirow{1}{*}{36.83} & \multirow{1}{*}{36.53} & \multirow{1}{*}{36.53} & \multirow{1}{*}{71.17} & \multirow{1}{*}{71.22} & \multirow{1}{*}{71.22} \\ \midrule
\multirow{1}{*}{6339.80} & \multirow{1}{*}{0.70} & \multirow{1}{*}{0.07} & \multirow{1}{*}{0.06} & \multirow{1}{*}{0.06} & \multirow{1}{*}{3.49} & \multirow{1}{*}{3.30} & \multirow{1}{*}{3.30} & \multirow{1}{*}{15.70} & \multirow{1}{*}{15.39} & \multirow{1}{*}{15.39} & \multirow{1}{*}{36.69} & \multirow{1}{*}{36.49} & \multirow{1}{*}{36.49} \\ \midrule
\multirow{1}{*}{6827.48} & \multirow{1}{*}{0.65} & \multirow{1}{*}{0.00} & \multirow{1}{*}{0.00} & \multirow{1}{*}{0.00} & \multirow{1}{*}{0.77} & \multirow{1}{*}{0.70} & \multirow{1}{*}{0.70} & \multirow{1}{*}{5.51} & \multirow{1}{*}{5.31} & \multirow{1}{*}{5.31} & \multirow{1}{*}{16.38} & \multirow{1}{*}{16.13} & \multirow{1}{*}{16.13} \\ \midrule
\multirow{1}{*}{7396.43} & \multirow{1}{*}{0.60} & \multirow{1}{*}{0.00} & \multirow{1}{*}{0.00} & \multirow{1}{*}{0.00} & \multirow{1}{*}{0.12} & \multirow{1}{*}{0.10} & \multirow{1}{*}{0.10} & \multirow{1}{*}{1.52} & \multirow{1}{*}{1.43} & \multirow{1}{*}{1.43} & \multirow{1}{*}{6.11} & \multirow{1}{*}{5.93} & \multirow{1}{*}{5.93} \\ \midrule
\multirow{1}{*}{8068.84} & \multirow{1}{*}{0.55} & \multirow{1}{*}{0.00} & \multirow{1}{*}{0.00} & \multirow{1}{*}{0.00} & \multirow{1}{*}{0.01} & \multirow{1}{*}{0.01} & \multirow{1}{*}{0.01} & \multirow{1}{*}{0.31} & \multirow{1}{*}{0.28} & \multirow{1}{*}{0.28} & \multirow{1}{*}{1.82} & \multirow{1}{*}{1.73} & \multirow{1}{*}{1.73} \\ \bottomrule
\end{tabular}%
%}
\footnotesize \ \ (\ref{eq:l21}): 12-point rule Composite Newton-cotes Quadrature \  \  (\ref{eq:l16}):  FRFT-CDF 
\vspace{-0.3cm}
\end{table}

\newpage	
																										
 \noindent
The Fractional Fourier Transform (FRFT) algorithm is used to perform the  GTS cumulative function in the extended Black-Scholes formula (\ref{eq:l16}), whereas the twelve-point rule Composite Newton--Cotes Quadrature algorithm is used to perform the Generalized Black-Scholes Formula (\ref{eq:l21}). As expected,  both algorithms produce the same estimation for European option prices at two decimal places.\\
 \noindent
To generalize the analysis and account for a large range of option moneyness and maturity, The error (\ref{eq:l27}) was computed as the difference between VG and BS option~prices:
\begin{align}
Error(k,\tau)= F^{GTS}_{call}(S_{t},\tau) - F^{BS}_{call}(S_{t},\tau)  \quad \quad  \hbox{$(k=\frac{S_{t}}{K})$} \label {eq:l27}
\end{align}
\noindent
Fig~\ref{fig1ab} displays the error ($Error(k,\tau)$) as a function of the time to maturity ($\tau$) and the option moneyness ($k$). The spot price ($S_{t}$) is a constant, and the option moneyness depends on the strike price. The Black-Scholes (BS) model and GTS European option produce different pricing results. Compared to the option price under the GTS distribution, the BS model is underpriced for the at-the-money (ATM) and the in-the-money (ITM) options as shown in Fig~\ref{fig1ab}. 
 \begin{figure}[ht]
\vspace{-0.5cm}
    \centering
\hspace{-0.6cm}
  \begin{subfigure}[b]{0.47\linewidth}
    \includegraphics[width=\linewidth]{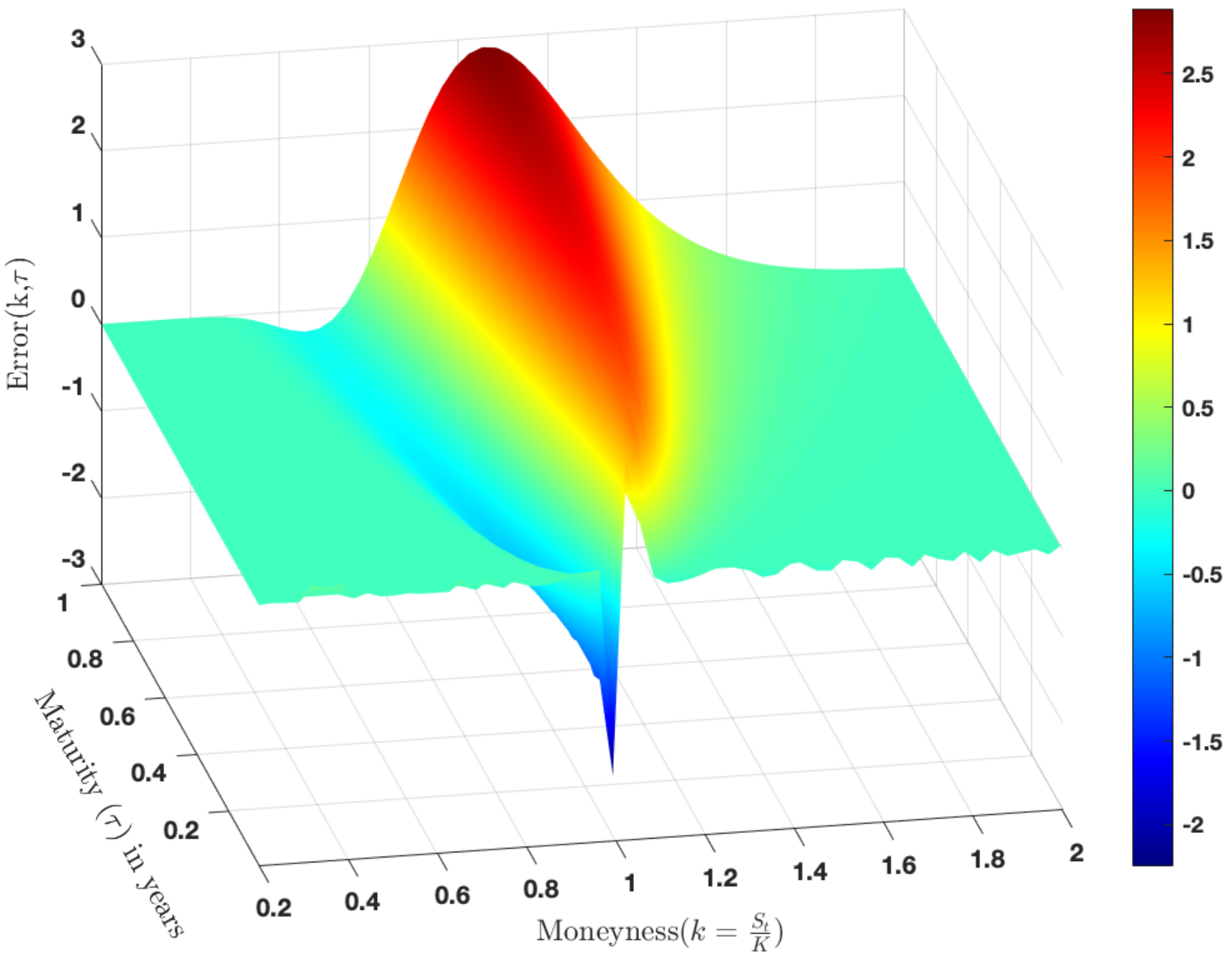}
\vspace{-0.3cm}
     \caption{Error(k,$\tau$)}
         \label{fig101a}
  \end{subfigure}
\hspace{-0.2cm}
  \begin{subfigure}[b]{0.47\linewidth}
    \includegraphics[width=\linewidth]{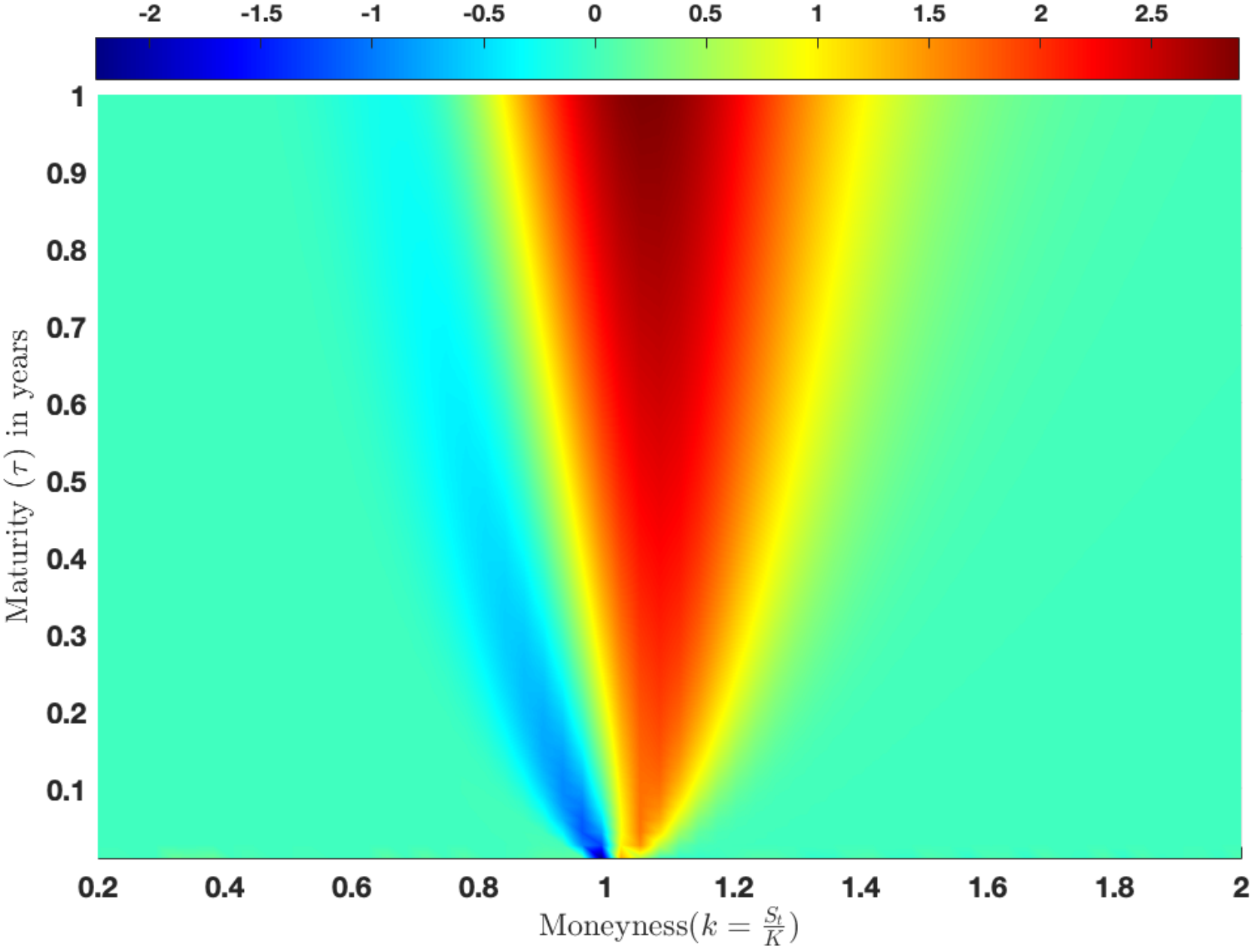}
\vspace{-0.5cm}
     \caption{Error(k,$\tau$) (top view)}
         \label{fig101b}
          \end{subfigure}
\vspace{-0.5cm}
  \caption{Combined Effects of time to Maturity and Option Moneyness}
  \label{fig1ab}
\vspace{-0.6cm}
\end{figure}

\noindent 
For the out-the-money (OTM) option, the Black-Scholes (BS) model is slightly overpriced (dark blue color in Fig~\ref{fig1ab}).  However, the BS model and GTS European options yield the same option price for the deep out-of-the-money (OTM) and the deep-in-the-money (ITM) options.\\
\noindent 
The findings can be compared to the European option pricing under the Variance-Gamma process \cite{nzokem2022, mozumder2015revisiting}. In both studies, we have the same pattern. the BS model is underpriced for the in-the-money (ITM) option, whereas the BS model is overpriced for the out-the-money (ITM) option. However, The error ($Error(k,\tau)$) between the BS model and GTS European options has a slightly different magnitude compared to the magnitude error($Error(k,\tau)$) between the BS model and option pricing under the Variance-Gamma process \cite{nzokem2022}.
%\clearpage
%\newpage
\section {Conclusion} 
\noindent 
The paper examines the pricing of the European option when the log asset price follows a rich class of Generalized Tempered Stable (GTS) distribution. The S\&P 500 data is used to illustrate the pricing results. The Equivalent Martingale Measure (EMM)  of the GTS distribution exists, and the Esscher transform method preserves the structure of the GTS distribution. \\
The extended Black-Scholes formula was computed based on the cumulative distribution function (CDF) generated by the Fractional Fast Fourier (FRFT) algorithms. In addition, The Generalized Black-Scholes Formula was computed using the 12-point rule Composite Newton-Cotes Quadrature algorithms. Both algorithms yield the same European option price at two decimal places. The Black-Scholes (BS) model and GTS European options produce different pricing results. Compared to the option price under the GTS distribution, the BS model is underpriced for the near-the-money (NTM) and the in-the-money (ITM) options.  However, the BS model and GTS European options yield the same option price for the deep out-of-the-money (OTM) and the deep-in-the-money (ITM) options.\\

\vspace{6pt}
%%%%%%%%%%%%%%%%%%%%%%%%%%%%%%%%%%%%%%%%%%
 \textbf{Funding Declaration:} this research received no funding.
%%%%%%%%%%%%%%%%%%%%%%%%%%%%%%%%%%%%%%%%%%
%\cleardoublepage
%\pagestyle{plain}
%\onehalfspacing
%\newpage
%\addcontentsline{toc}{chapter}{References}
%\input{referenc}
%\bibliographystyle{IEEEtran}
%\bibliographystyle{unsrt}
%\bibliographystyle{plain}
%\nocite{*}
\setlength{\bibsep}{0pt plus 0.3ex}
\bibliography{nzokem_pricing}

\end{document}